\documentclass[conference]{IEEEtran}
\IEEEoverridecommandlockouts
\input{commands.tex}
\usepackage{cite}
\usepackage{tikz}
\usepackage{balance}
\usetikzlibrary{positioning}
\RequirePackage{tikz}
\usetikzlibrary{
	angles,
	arrows,
	automata,
	backgrounds,
	calc,
	cd,
	decorations.markings,
	decorations.pathreplacing,	
	calligraphy,
	fit, 
	positioning,
	shapes, 
	shapes.multipart,
	matrix, 
	quotes
}
\tikzset{
	comp/.style={align=center, shape=circle, draw=black, thick, minimum width=1cm, minimum height=0.5cm, initial text = {}},
	every edge/.style= {->, draw = black, >=stealth, thick, shorten >= 2pt},
	 mat/.style={matrix of nodes, inner sep=0pt,
		nodes={text depth=0.8ex, text height=1em, minimum width=6ex,
			inner ysep=1pt, inner xsep=4pt, outer sep=0pt, anchor=west,rectangle,draw=black,align=center},
		column sep=-\pgflinewidth,
		row sep= -\pgflinewidth}
}

\usepackage{amsmath,amssymb,amsfonts}
\usepackage{algorithmic}
\usepackage{graphicx}
\usepackage{textcomp}
\usepackage{xcolor}
\def\BibTeX{{\rm B\kern-.05em{\sc i\kern-.025em b}\kern-.08em
    T\kern-.1667em\lower.7ex\hbox{E}\kern-.125emX}}
\usepackage{ltl}

\begin{document}

\title{Canonical Representations of \\$k$-Safety Hyperproperties
  \thanks{This work was partially supported by the German Research Foundation (DFG) as part of the Collaborative Research Center ``Methods and Tools for Understanding and Controlling Privacy'' (CRC 1223) and the Collaborative Research Center ``Foundations of Perspicuous Software Systems'' (TRR 248, 389792660), and by the European Research Council (ERC) Grant OSARES (No. 683300).}}

\author{\IEEEauthorblockN{Bernd Finkbeiner}
\IEEEauthorblockA{\textit{Reactive Systems Group} \\
\textit{Saarland University}\\
Saarbr\"ucken, Germany \\
finkbeiner@react.uni-saarland.de}
\and
\IEEEauthorblockN{Lennart Haas}
\IEEEauthorblockA{\textit{Graduate School of Computer Science} \\
\textit{Saarland University}\\
Saarbr\"ucken, Germany \\
lennart.haas@stud.uni-saarland.de}
\and
\IEEEauthorblockN{Hazem Torfah}
\IEEEauthorblockA{\textit{Reactive Systems Group} \\
\textit{Saarland University}\\
Saarbr\"ucken, Germany \\
torfah@react.uni-saarland.de}
}
\maketitle

\begin{abstract}
Hyperproperties elevate the traditional view of trace properties form sets of traces to sets of sets of traces and provide a formalism for expressing information-flow policies. 
For trace properties, algorithms for verification, monitoring, and synthesis are typically based on a representation of the properties as omega-automata. For hyperproperties, a similar, canonical automata-theoretic representation is, so far, missing. This is a serious obstacle for the development of algorithms, because basic constructions, such as learning algorithms, cannot be applied.

In this paper, we present a canonical representation for the widely used class of regular $k$-safety hyperproperties, which includes important polices such as noninterference.  We show that a  regular $k$-safety hyperproperty $\safety$ can be represented by a finite automaton, where each word accepted by the automaton represents a violation of~$\safety$.  
The representation provides an automata-theoretic approach to regular $k$-safety hyperproperties and allows us to compare regular $k$-safety hyperproperties, simplify them, and learn such hyperproperties. We investigate the problem of constructing automata for  regular $k$-safety hyperproperties  in general and from formulas in \hyperltl, and provide complexity bounds for the different translations. We also present a learning algorithm for regular $k$-safety hyperproperties based on the L$^*$ learning algorithm for deterministic finite automata. 

\end{abstract}

\begin{IEEEkeywords}
Hyperproperties,
Automata,
Learning,
Information-flow control.
\end{IEEEkeywords}

\section{Introduction}

Hyperproperties~\cite{clarkson_et_al:hyperproperties} generalize traces properties to sets of sets
of traces.  Famous examples of hyperproperties that cannot be expressed as trace properties are information-flow policies, such as noninterference, because they relate
multiple runs of a system: a violation of a information-flow policy
can therefore only be detected by looking at trace sets with more than one trace.

Many verification and analysis techniques for trace properties are,
nowadays, based on the \emph{automata-theoretic
approach}~\cite{Vardi87verificationof}, whereby the property is
translated into an equivalent automaton and then processed by standard
operations on automata.
For hyperproperties, there is, so far, no automata-theoretic
foundation. This means that algorithms for hyperproperties cannot
be based directly on automata transformations.

One might argue that the lack of an automata representation is not a big issue, because many
verification problems, such as model checking against $k$-safety 
hyperproperties, can be reduced, via a self-composition of the system under verification, to
standard trace-based model checking against a trace property.
However, there are important algorithmic approaches that do not
translate this easily. A prime example are learning algorithms like
Dana Angluin's L$^*$
algorithm~\cite{angluin_learning_regular_sets}. Learning is a fundamental building block for compositional
verification~\cite{cobleigh_et_al:learning_assumptions_for_compositional_verification}, synthesis~\cite{skeletons}, and for 
mining specifications of malicious behavior~\cite{vanLamsweerde:1998:IDR:297569.297578,Fern:2004:LDC:3037008.3037033,Cobleigh:2006:UGC:1181775.1181801}. Generally, the advantage of learning algorithmis like L$^*$
compared to other construction methods is that the number of queries
the learner needs to pose to the teacher is determined by the size of
the smallest deterministic automaton for the target language. Usually,
this is significantly smaller than the intermediate automata
that
occur in a direct construction.

In this paper, we develop an automata representation for the class of regular $k$-safety hyperproperties.  The $k$-safety hyperproperties are those
hyperproperties where every set of traces that violates the
hyperproperty contains a set of at most $k$ bad trace prefixes, such
that every extension of the bad prefixes also violates the
hyperproperty.  We represent a $k$-safety hyperproperty using a bad-prefix automaton, a
finite-word automaton that recognizes the bad prefixes as finite words
over an alphabet consisting of $k$-tuples, where each word in the
language is interpreted as a set of (at most) $k$ traces.  A
$k$-safety hyperproperty may, in principle, have many different
representations as such a bad prefix language. Consider, for example, the $2$-safety hyperproperty given by the \hyperltl formula $\varphi=\forall \pi \forall \pi'.~\globally (a_\pi \rightarrow a_{\pi'})$ over the set of atomic propositions $\{a\}$, which specifies for each pair of traces $\pi, \pi'$, that whenever $a$ holds on $\pi$ it also holds on $\pi'$. A bad prefix for $\varphi$ is, for example, the set of finite traces $\{t,t'\}$ where $t = \{a\}\{a\}$ and $t'=\{a\}\{\}$. A tuple representation of $\{t,t'\}$ is the sequence $(\{a\},\{a\})(\{a\},\{\})$. Since the set defines no order on $t$ and $t'$, another representation of the bad prefix is the sequence $(\{a\},\{a\})(\{\},\{a\})$.  
  
Just as for bad prefixes for trace properties, the bad prefixes may or may not be
minimal; additionally, any ordering of traces in a trace set will lead
to a different tuple representation. Using the terminology for
trace properties~\cite{kupferman_et_al:model_checking_safety_properties}, we define a bad-prefix automaton as \emph{tight} if it accepts all  bad prefixes; additionally, we say the automaton is
\emph{permutation-complete} if it is closed under permutations of the
tuples. Minimal deterministic bad-prefix automata that are both tight
and permutation-complete provide a canonical representation for
$k$-safety hyperproperties. We provide algorithms for constructing permutation-complete bad-prefix automata for regular $k$-safety hyperproperties starting from representations in \hyperltl, nondeterministic bad-prefix automata and deterministic bad-prefix automata.

Based on this automaton representation, we present the first
\emph{learning algorithm for hyperproperties}.  Our algorithm learns a minimal deterministic tight permutation-complete  bad-prefix automaton for
some unknown regular $k$-safety regular hyperproperty and an unknown minimal $k$.

The remainder of the paper is structured as follows.  We give
background on hyperproperties and automata in
Section~\ref{sec:Background}.  Section~\ref{sec:automata} introduces
automata for $k$-safety hyperproperties and establishes basic facts
about tight and permutation-complete bad-prefix automata.
In Section~\ref{sec:learning} we present a learning framework for learning $k$-safety regular hyperproperties and a realization of the framework for \hyperltl in Section~\ref{sec:LearningHyperLTL}. With Section~\ref{sec:learnability} we conclude with some decidability results on the learnability of $k$-safety-hyperproperties.

\section{Background}
\label{sec:Background}
\subsection{Hyperproperties.} 
A \emph{trace property} $\T$ over an alphabet $\Sigma$ is a set of infinite traces from $\Sigma^\omega$. A trace $t \in \Sigma^\omega$ satisfies the property $\T$ if $t \in T$.  The set of all trace properties over the alphabet $\Sigma$ is denoted by $\mc{P}(\Sigma^{\omega})$.
 
 A \emph{hyperproperty}  over an alphabet $\Sigma$ is a set $\H  \subseteq \mc{P}(\Sigma^\omega)$ of sets of infinite traces over $\Sigma$~\cite{clarkson_et_al:hyperproperties}. 
A set of infinite traces $\T \subseteq \Sigma^\omega$  satisfies a hyperproperty $\H$ if $\T \in \H$. 

\subsection{\hyperltl: A temporal logic for  hyperproperties.}  
Let $\mathcal{V}$ be an infinite supply of trace variables and let $\AP$ be a set of atomic propositions.  The syntax of HyperLTL is given by the following grammar:
\begin{align*}
	\psi~&::=~\exists \pi.\;\psi~~|~~\forall\pi.\;\psi~~|~~\varphi \\
	\varphi~&::=~a_{\pi}~~|~~\neg \varphi~~|~~\varphi \vee \varphi~~|~~\lnext \varphi~~|~~\varphi\, \until \varphi 
\end{align*}
where $a\in \AP$ is an atomic proposition and $\pi \in \mathcal V$ is a trace variable. Note that atomic propositions are indexed by trace variables.
The quantification over traces makes it possible to express properties like ``on all traces $\psi$ must hold'', which is expressed by $\forall \pi.~\psi$. Dually, one can express that ``there exists a trace such that $\psi$ holds'', which is denoted by $\exists \pi.~\psi$. 
The temporal operators are defined as for \ltl. The next operator $\LTLcircle \psi$ states that the next step along a trace must satisfy $\psi$. The until operator $\psi_1 \LTLuntil \psi_2$ states that $\psi_1$ must hold along a trace until $\psi_2$ holds. 
We also use the derived temporal operators $\LTLdiamond \psi$ eventually $\psi$ holds, $\LTLsquare \psi$ the formula $\psi$ holds on all trace positions, and $\psi_1 \LTLrelease \psi_2$, the release operator, the dual to $\LTLuntil$, that states that $\psi_2$ may not hold only after $\psi_1$ has been fulfilled otherwise $\psi_2$ must hold forever.

We abbreviate the formula $\bigwedge_{x\in X} (x_\pi \leftrightarrow x_{\pi'})$, expressing that the traces $\pi$ and $\pi'$ are equal with respect to a set of atomic propositions $X \subseteq \AP$  by $\pi =_X \pi'$.


\begin{example}
\label{exmpl:noninterference}
	The following \hyperltl formula defines the security policy of reactive noninterference. Let $\AP = I \cup O$, where  $I$ and $O$ are sets of low-security inputs and low-security outputs, respectively: 
	$$\forall \pi. \forall \pi'.~(\pi \not =_I \pi') \release\, (\pi =_O \pi')$$
	The formula states that, for every pair of traces, as long as there is no difference in the observed inputs, no difference should be observed in the outputs.
\end{example}

Let $\T$ be a set of traces of some alphabet $2^\AP$ for some set of atomic propositions \AP. Formally, the semantics of HyperLTL formulas is given with respect to a \emph{trace assignment} $\Pi$ from $\mathcal{V}$ to $\T$, i.e., a partial function mapping trace variables to actual traces. $\Pi[\pi \mapsto t]$ denotes that $\pi$ is mapped to $t$, with everything else mapped according to $\Pi$. For a trace $t \in \T$, let $t[i,\infty]$ denote the suffix of $t$ starting at position $i$. With $\Pi[i,\infty]$ we denote the trace assignment that is equal to $\Pi(\pi)[i,\infty]$ for all~$\pi$.
\begin{align*}
	&\Pi \models_T~\exists \pi. \psi &&\text{iff}\hspace{5ex} \exists~t \in T~:~ \Pi[\pi \mapsto t] \models_T \psi \\
	&\Pi \models_T~\forall \pi. \psi &&\text{iff}\hspace{5ex} \forall~t \in T~:~ \Pi[\pi \mapsto t] \models_T \psi \\
	&\Pi \models_T~a_{\pi} &&\text{iff}\hspace{5ex} a \in \Pi(\pi)[0] \\
	&\Pi \models_T~\neg \psi &&\text{iff}\hspace{5ex} \Pi \not \models_T \psi \\
	&\Pi \models_T~\psi_1 \vee \psi_2 &&\text{iff}\hspace{5ex} \Pi \models_T \psi_1~\text{or}~\Pi \models_T \psi_2 \\
	&\Pi \models_T~\lnext \psi &&\text{iff}\hspace{5ex} \Pi[1,\infty] \models_T \psi \\
	&\Pi \models_T~\psi_1 \until \psi_2 &&\text{iff}\hspace{5ex} \exists~i \geq 0 : \Pi[i,\infty] \models_T \psi_2 \\
	& &&\hspace{7ex} \wedge \forall~0 \leq j < i.~\Pi[j,\infty] \models_T \psi_1
\end{align*}
We say a set of traces $T$ \emph{satisfies} a HyperLTL formula $\varphi$ if $\Pi \models_T \varphi$, where $\Pi$ is the empty trace assignment.

We call a \hyperltl formula $\varphi$ syntactically-safe if it is of the form $\varphi =\forall^*\pi.\psi$ and $\psi$ is a syntactically-safe LTL formula, i.e., an LTL formula where the only temporal operators are $\lnext$ and $\release$.

\subsection{Automata.} A nondeterministic finite automaton (NFA) is defined as a tuple $\A = (Q,\Sigma, q_0, F, \delta)$, where $Q$ denotes a finite set of states, $\Sigma$ denotes a finite alphabet, $q_0$ denotes a designated initial state, $F\subseteq Q$ denotes the set of accepting states, and $\delta: Q \times \Sigma \rightarrow \mc{P}(Q)$ denotes the transition relation that maps a state and a letter to the set of successor states. 
A run in a $\A$ on a finite word $w = w_0 \dots w_n \in \Sigma^*$ is a sequence of states $ r = q_0\dots q_{n+1} \in Q^*$ with $q_{i+1} \in \delta(q_i,w_i)$ for all $0 \le i \le n$. The run $r$ is accepting if $q_{n+1} \in F$. The set of all accepted words by an automaton $\A$ is called its language and is denoted by $\L(\A)$. The size of an automaton is the size of its set of states $Q$ and is denoted by $|\A|$.

Deterministic finite automata (DFA) are a special case of NFAs, where $\abs{\delta(q,a)} \le 1$ for all $q \in Q$ and $a \in \Sigma$. The transition relation of a deterministic automaton can be given as a function $\delta: Q \times \Sigma \rightarrow Q$. 

A B\"uchi automaton $\mathcal B = (Q,\Sigma, q_0, F, \Delta)$ is an automaton over infinite words. A run of $\mathcal B$ on an infinite word $w = w_1w_2 \dots \in \Sigma^{\omega}$  is an infinite sequence $ r = q_0q_1\dots \in Q^{\omega}$ with $q_{i+1} \in \delta(q_i,w_i)$ for all $i \in \NN$. A run $r$ is accepting if there exist infinitely many $i \in \NN$ such that $q_i \in F$. 
A B\"uchi automaton $\A = (Q,\Sigma, q_0, F, \Delta)$ is called safety automaton if $Q = F$, i.e., every run on a safety automaton is accepted. In the rest of the paper we omit the set $F$ from the tuple representation of safety automata. 

\subsection{Safety Languages.} A finite word $w = w_1 \dots w_i \in \Sigma^*$ is called a bad prefix for a language $L\subseteq\Sigma^{\omega}$, if every infinite word $v \in \Sigma^{\omega}$ with prefix $w$ is not in the language $L$. A language $L \subseteq \Sigma^{\omega}$ is called a safety language if every $w\not \in L$ has  a  bad prefix for $L$. We denote the set of all bad prefixes for a language $\L$ by $\bad(L)$. We say $X \subseteq \bad(L)$ is a trap for $L$, if for every $w \not \in L$, there exists a prefix of $w$ in $X$ and denote the set of all traps by $\textsf{Trap}(L)$.

For every $\omega$-regular safety language $L$, 
a finite automaton $\A$ that accepts the bad prefixes of $L$ is called a bad-prefix automaton for $L$. We say that $\A$ is tight if $L(\A) = \bad(L)$ and fine if there exists some $X \in \textsf{Trap}(L)$ and $L(\A) = X$. \\


\subsection{Notations.} For a sequence $t = \alpha_1\alpha_2\dots$ and $i\leq j \in \NN$, $t[i]=\alpha_i$, $t[i,j] = \alpha_i \dots \alpha_j$. For $t \in \Sigma^{\omega}$, $t[i, \infty] = \alpha_i\alpha_{i+1}\dots$. 

For  $t \in \Sigma^*$ and $\tau \in \Sigma^* \cup \Sigma^\omega$, $t$ is a prefix of $\tau$ denoted by $t \le \tau$ if and only if $\abs{t} \le \abs{\tau} \wedge \forall i \le \abs{t}. \ t[i] = \tau[i]$.

\section{Automata for $k$-Safety-Hyperproperties}\label{sec:HyperAutomata}
\label{sec:automata}
\subsection{Representations of $k$-Safety-Hyperproperties}

\begin{figure*}[ht]
\begin{align*}
& \safety = \forall \pi.\forall\pi'.\  ( {\pi}\not =_I {\pi'})\release  (\pi = _O {\pi'})
\\ &\bad(\safety) = \{T \subseteq \Sigma^* \mid \exists t,t' \in T.~\exists j.~t[...j]_I = t'[...j]_I \wedge t[j]_O \not = t'[j]_O   \}
\\  &\bad(\safety,2) = \{\{t,t'\} \subseteq \Sigma^*  \mid \exists j.~t[...j]_I = t'[...j]_I \wedge t[j]_O \not = t'[j]_O  \}
\\ &\text{Rep}_1 = \{ (\alpha_0,\alpha'_0) \dots(\alpha_m,\alpha'_m) \in (\Sigma^2)^* \mid  \forall j. (\alpha_j)_I = (\alpha'_j)_I ~\wedge~ \exists i. \exists o \in O.~ o \in \alpha_i \wedge o \not \in \alpha'_i \}
\\ &\text{Rep}_2 = \{  (\alpha_0,\alpha'_0) \dots(\alpha_m,\alpha'_m) \in (\Sigma^2)^* \mid  \forall j. (\alpha_j)_I = (\alpha'_j)_I ~\wedge~ \exists i. \exists o \in O.~ o \not\in \alpha_i \wedge o  \in \alpha'_i  \}
\\ &\per{\safety,2} = \{ (\alpha_0,\alpha'_0) \dots(\alpha_m,\alpha'_m) \in (\Sigma^2)^* \mid  \forall j. (\alpha_j)_I = (\alpha'_j)_I ~\wedge~ \exists i. \exists o \in O.~ o \in \alpha_i \leftrightarrow o \not \in \alpha'_i  \}
\end{align*}
\caption{A 2-safety-hyperproperty given by a \hyperltl formula $\safety$. The formula $\safety$ defines the information flow policy of reactive noninterference.  The sets $\bad(\safety)$, $\bad(\safety,2)$, $\text{Rep}_1$, $\text{Rep}_2$, and $\per{\safety,2}$ define the sets of bad-prefixes, 2-bad-prefixes, two different 2-representations, and the set of all 2-representation of $\safety$, respectively. The set $\Sigma $ is defined as $\Sigma = 2^\AP$ for a set of atomic propositions $\AP= O \cup I $.}
\label{fig:representations}
\end{figure*}

The definition of safety can be generalized to hyperproperties by generalizing the definition of bad-prefixes from a finite trace to a finite set of finite traces \cite{clarkson_et_al:hyperproperties}.
For a set of finite traces $T \subseteq \Sigma^*$ and a set of infinite traces $T' \subseteq \Sigma^{\omega}$, we say that $T$ is a prefix of $T'$, denoted by $T \le T'$, if and only if $\forall t \in T. \exists t' \in T'.\ t \le t'$.
A hyperproperty $\safety$ over $\Sigma$ is \emph{hypersafety} if and only if
\begin{align*}
\forall T' \subseteq \Sigma^\omega.\ (T' \not\in \safety \Rightarrow &\exists T \subseteq \Sigma^*.\ ( T \le T' \wedge
\\ &\forall \widetilde{T} \subseteq \Sigma^\omega. \ ( T \le \widetilde{T} \Rightarrow \widetilde{T} \not \in \safety)))
\end{align*}
We call $T$ a \emph{bad-prefix} for the hyperproperty $\safety$. We denote  the set of bad-prefixes for a hypersafety property $\safety$ by $\bad(\safety) = \{	 T \subseteq \Sigma^* \mid \forall T' \subseteq\Sigma^\omega.\ (T \leq  T' \Rightarrow T' \not \in \safety) \}$.  We call a bad prefix $T$ for $\safety$ minimal, if and only if, there exists no $T' < T $ that is also a bad prefix for $\safety$. 

\begin{definition}[$k$-safety hyperproperty]
	For any $k' \in \nats$, let $\bad(\safety,k') = \{ T \in \bad(\safety) \mid \abs{T} \le k' \}$ . We call an element of $\bad(\safety,k')$ a \emph{$k'$-bad-prefix} for $\safety$.
 A safety hyperproperty $\safety$ is a \emph{$k$-safety hyperproperty}, if every set $T'\not\in\safety$ has a $k$-bad-prefix. 
 \end{definition}

In the next section, we define finite automata for $k$-safety hyperproperties by defining automata that represent their sets of bad-prefixes. Each finite bad prefix of a safety-hyperproperty can be represented by a finite word as follows. 

\begin{definition}[Representations of $k$-safety hyperproperties]
	For a sequence $\sigma = \vv{v}_0\vv{v}_1 \vv{v}_2\dots \vv{v}_m \in (\Sigma^k)^*$, let  $\unzip$ be the mapping defined as $\unzip(\sigma) = \{t_i \in \Sigma^*\mid 1 \leq i \leq k, \forall 0 \leq j\leq m .~ t_i[j] = \vv{v}_j[i]\}$.  We call $\sigma \in (\Sigma^k)^*$ a \emph{representation} of $T \subseteq \Sigma^*$ if  \unzip$(\sigma) = T$.
	
	For a $k$-safety-hyperproperty $\safety$,
	 a language $L \subseteq (\Sigma^{k'})^*$ is called a \emph{representation} of $\safety$ for some $k'\in \nats$, when:
	 for all $T \subseteq \Sigma^\omega$, $T\not\in \safety$,  if and only if, there exists $\sigma \in L$, such that, $\unzip(\sigma) \subseteq T$ and $\unzip(\sigma) \in \bad(\safety)$.  We call $k'$ the \emph{arity} of the representation and further call $L$ a \emph{$k'$-representation} of $\safety$. 
\end{definition}

	%

%

We extend the definition of \unzip\ to languages. For a language $L \subseteq (\Sigma^k)^*$, $\unzip(L) = \{\unzip(\sigma) \mid \sigma \in L\}$. 

Notice that a $k$-safety-hyperproperty has several  representations of different arities. It also has several representations of the same arity (by permuting the order on the traces). We denote the set that defines the union of all representations of a $k$-safety hyperproperty $\safety$ of arity $k'$ by $\per{\safety,k'}$.

\begin{example}
	The security policy of reactive noninterference given in Example~\ref{exmpl:noninterference} is an example of a 2-safety hyperproperty. In Figure~\ref{fig:representations} the policy is given by the \hyperltl formula $\safety$. 
	A violation of $\safety$ along two traces is observed, if up to some position, the traces share the  same input sequence and differ in the output values at this position. The set of bad-prefixes for $\safety$ is given by the set $\bad(\safety)$. 
	To check whether there is a violation of $\safety$ it is sufficient to find two traces that violate $\safety$, i.e.,  any set of traces that violates $\safety$ has a  bad prefix of size two. The set $\bad(\safety,2)$ gives all the bad-prefixes of size two.
	 Two sets of 2-representations of these bad-prefixes are given by the sets $\text{Rep}_1$ and $\text{Rep}_2$\footnote{The sets $\text{Rep}_1$ and $\text{Rep}_2$ are not the only sets with 2-representations of the bad-prefixes of $\safety$. For $\safety$ there is an infinite number of distinct 2-representations.}. To understand the difference between the representations in $\text{Rep}_1$ and $\text{Rep}_2$ look at the following two traces:
	  Assume w.l.o.g. that $I=\{i\}$ and $O=\{o\}$ and let $t=\{i,o\}\{i,o\}\{i,o\} \dots $ and $t'=\{i,o\}\{i,o\}\{i\} \dots $ be two infinite traces over $2^{O \cup I}$. The set $\{t,t'\}$ violates $\safety$ with the bad prefix $T=\{~\{i,o\}\{i,o\}\{i,o\} ~,~ \{i,o\}\{i,o\}\{i\}~\}$. The set $T$ has two 2-representations: the sequence  $\sigma_1=(\{i,o\},\{i,o\})(\{i,o\},\{i,o\})(\{i,o\},\{i\})$, which is in the set $\text{Rep}_1$ but not in  $\text{Rep}_2$, and another representation is $\sigma_2 = (\{i,o\},\{i,o\})(\{i,o\},\{i,o\})(\{i\},\{i,o\})$ which belongs to   $\text{Rep}_2$ but not to  $\text{Rep}_1$. Both $\sigma_1$ and $\sigma_2$ belong, however, to the set $\per{\safety,2}$, which contains all representations of 2-bad-prefixes of $\safety$. 
	
\end{example}
In general, for any $k$-safety hyperproperty $\safety$, if a sequence $\sigma\in\per{\safety, k'}$ for any $k'\in \nats$, then so is any permutation of~$\sigma$.

\begin{theorem}
	For every $k$-safety hyperproperty $\safety$, and for $k'\ge k$, there is a $k'$-representation of $\safety$.  
\end{theorem}
\begin{proof}Clearly, every $k$-safety hyperproperty has a representation of arity $k$. 
	Let $L$ be a $k$-representation for $\safety$. Define $L'$ such that each $\sigma' \in L'$ is of the form $\sigma' =(\alpha_0^1,\dots,\alpha_0^k,\dots,\alpha_0^{k'})(\alpha_1^1, \dots, \alpha^k_1,\dots,\alpha_1^{k'})\dots \in (\Sigma^{k'})^*$, where $ (\alpha^1_0,\dots, \alpha_0^k)(\alpha_1^1, \dots, \alpha^k_1)\dots \in L$, and for all $i \in \nats$ and for all $k<j\leq k'$ we have $\alpha_i^j = \alpha_i^k$. Let the set $\unzip(\sigma')= \{t_1,\dots,t_k,\dots,t_{k'} \}$. Clearly, for $k< j\leq k'$, we have $t_j= t_k$. Thus, $\unzip(L') = \unzip(L)$, which makes $L'$ a $k'$-representation of $\safety$. 
\end{proof}

In the rest of the paper, the length of a bad prefix $T$ is the length of the longest trace in $T$. The size of a bad prefix $T$ is the size $|T|$. 	

\subsection{Bad-prefix automata for $k$-safety hyperproperties}
We now develop a canonical representation for $k$-safety hyperproperties. We start by defining bad-prefix automata for $k$-safety hyperproperties. At the end of the section we show that \emph{minimal, deterministic, tight and permutation-complete} bad-prefix automata  give a canonical representation for $k$-safety hyperproperties.   
\begin{definition}[Regular $k$-safety hyperproperties]
	A $k$-safety hyperproperty $\safety$ is called regular if a representation of  $\safety$ is a regular language. 
\end{definition}

If a  $k$-safety hyperproperty $\safety$ is regular, we can build an automaton that recognizes one of its representations for some arity $k'$. We call such an automaton a \emph{$k'$-bad-prefix automaton} for $\safety$. An automaton is a bad-prefix automaton for $\safety$, if it is a $k'$-bad-prefix automaton for some arity $k'\in \nats$. In the following, we show that we can distinguish different types of bad-prefix automata for $k$-safety-hyperproprties. The distinction is based on the traditional notions of \emph{tightness} and \emph{fineness} for bad-prefix automata for regular properties \cite{kupferman_et_al:model_checking_safety_properties}, and the novel notion of \emph{permutation-completeness} that we define later in this section. 

A \emph{tight} bad-prefix automaton for an $\omega$-regular property $\T$ accepts all bad-prefixes of $\T$. The language of a \emph{fine} bad-prefix automaton for $\T$ includes at least one bad prefix for each word $\sigma \not \in \T$. 
Following this tradition we can also make a similar distinction for bad-prefix automata for $k$-safety-hyperproperty~$\safety$. 
 
 \begin{definition}[Tight and fine $k$-bad-prefix automata]
 Let $\A$  be a $k'$-bad-prefix automaton for a $k$-safety-hyperproperty $\safety$ for some $k,k' \in \nats$. We call $\A$ \emph{tight} if and only if $A$ accepts a representation for each bad prefixes $T$ of $\safety$ with $|T|\leq k'$. 
 
  $\A$  is called \emph{fine}  if and only if for every word  $T \not \in\safety$ it accepts a representation of at least one  bad prefix (not necessarily the minimal one) of $T$.	
 \end{definition}


%



 Kupfermann and Vardi showed how to construct tight bad-prefix automata for safety-properties \cite{kupferman_et_al:model_checking_safety_properties}. 
The same constructions cannot be adapted for $k$-safety hyperproperties, due to the following reasoning. From its definition, a bad-prefix automaton~$\A$ for a $k$-safety-hyperproperty $\safety$ that is fine must,  for each set $T$ not in $\safety$, accept at least one representation of a bad prefix of $T$. 
If $\A$ is not tight then  either (1) $\A$ is not \emph{vertically tight}:accepts a representation for a bad prefix $T$,  but does not accept any representation for some $T\subset T'$ with $|T'|\leq k'$ which is also a bad prefix for $\safety$  or (2) $\A$ is not \emph{horizontally tight}: there is a representation $t'$ of a set $\{w \mid \exists w' \in T, w < w'\}$ that represents a smaller bad prefix for $\safety$, i.e., a trace in $T$  is not minimal. The latter case defines tightness according to the traditional definition as in \cite{kupferman_et_al:model_checking_safety_properties}.

 \remark{Notice that there exists no fine automaton that accepts a representation for a bad prefix $T$ that is not minimal, but accepts no representation for all bad prefixes $T'\subset T$. Assume that no representation of any $T'$ is accepted by $\A$. This means that there is word $\safety$ for which no representation of any of its bad-prefixes is accepted by $\A$,  namely the set $\widetilde{T}$, where each word in $\widetilde{T}$ is an infinite extension of a word in $T'$. This contradicts the assumption that $\A$ is a bad-prefix automaton for $\safety$. }

The next theorem how to construct a bad-prefix automaton that is horizontally tight using the construction presented in \cite{kupferman_et_al:model_checking_safety_properties}. A construction for vertical tightness is left for the theorem that follows. 
\begin{theorem}
For a $k$-safety hyperproperty $\safety$ over $\Sigma$, we can construct a tight bad-prefix automaton for $\safety$ of size:
\begin{itemize}
	\item $O(|\A|)$, when $\safety$ is represented by a deterministic bad-prefix automaton $\A$.
	\item $2^{O(|\A|)}$, when $\safety$ is represented by a  nondeterministic bad-prefix automaton $\A$.
\end{itemize}
\label{theo:constrtight}
\end{theorem}
\begin{proof}
The proof uses the ideas presented in \cite{fineautomata}.
	\begin{itemize}
		\item Let $\A=(Q,\Sigma^{k'}, q_0, F, \delta)$ be a deterministic bad-prefix automaton for $\safety$ for some $k'\ge k$. To construct a deterministic horizontally tight bad-prefix automaton for $\safety$ we replace the set of accepting states $F$ of $\A$ by a set $F'$ which is defined as follows:
		$$F' = \{ q\in Q \mid \forall \sigma \in Q^\omega.~q< \sigma \rightarrow \exists i\in \nats. \sigma[i]\in F \}$$
		The set $F'$ defines the set of states $q$ from which there is no infinite run in the automaton that has no accepting state. 
		\item If $\A$ is a nondeterministic bad-prefix automaton for $\safety$, we can construct an equivalent deterministic $k$-bad-prefix automaton $\A'$ of size $2^{|A|}$ and use the construction above. 
	\end{itemize}
\end{proof}

Bad-prefix automata for $k$-safety hyperproperties can also be distinguished according to the representations they accept. 
 A $k$-bad-prefix automaton $\A$ is called \emph{permutation-complete} if it accepts  all representations of every $k$-bad-prefix it accepts.

In general, the goal is to build a tight and permutation-complete bad-prefix automaton for a $k$-safety hyperproperty.  For tasks such as monitoring a system against a $k$-safety hyperproperty, such automata are of major importance.  With tight automata violations are detected as early as possible. A permutation-complete automaton does not depend on the ordering of the traces and therefore detects a violation no matter in what order the traces are observed.

In the next theorem we show how to construct a permutation-complete and tight $k$-bad-prefix automaton for a $k$-safety hyperproperty.

\begin{theorem}
For a deterministic $k$-bad-prefix automaton $\A$ of some $k$-safety hyperproperty $\safety$ over $\Sigma$, we can construct a deterministic, tight and permutation-complete $k$-bad-prefix automaton  of size $|A|^{k^{k}}$.
\label{theo:constrpermcomplete}
\end{theorem}
\begin{proof}
	 Let $\A = (Q,\Sigma^k, q_0,F,\delta)$ be a deterministic $k$-bad-prefix automaton for $\safety$. We construct a permutation-complete and vertically tight automaton $\A_{\mathfrak P}$ for $\safety$ that accepts a word $\sigma \in (\Sigma^k)^*$ if any of its permutations \footnote{From now on, if not stated otherwise, we use the word permutation to mean permutation with repetition.}is accepted by $\A$. We define these permutations as follows. 
		Let $\varsigma_1, \dots, \varsigma_{k^k}:\{1,\dots,k\} \rightarrow \{1, \dots, k\}$ be pairwise different functions. 
		A permutation of a tuple $(t_1, \dots, t_k)$ with respect to one function $\varsigma_i$ for $1 \leq i \leq k^k$ is a tuple $(t_{\varsigma_i(1)}, \dots, t_{\varsigma_i(k)})$. 
		The deterministic bad-prefix automaton $\A_{\mathfrak P}$ is defined by the tuple $(Q_{\mathfrak P},\Sigma^k, q_{0,\mathfrak P}, F_{\mathfrak P}, \delta_{\mathfrak P})$, where:
			\begin{itemize}
				\item $Q_{\mathfrak P} = (Q_1\times \dots \times Q_{k^k})$ where $Q_i = \{(q,i) \mid q\in Q\}$ for $1\leq i \leq k^k$. A set of states $Q_i$ resembles a copy of the automaton $\A$ that accepts a word $\sigma$ if it is a permutation of a word $\sigma'$ accepted by $\A$ with respect to the permutation function $\varsigma_i$.  The initial state $q_{0,\mathfrak P} = ((q_0,1),\dots,(q_{0},k^k))$.
				\item A word is accepted if one of its permutations is accepted. We define the set of accepting states as $F_{\mathfrak P} = \{((q_1,1), \dots, (q_{k^k},k^k)) \mid \exists i.~q_i \in F \}$.
				\item The transition relation $\delta_{\mathfrak P}$ is defined as follows:
					$$((q_1,1),\dots,(q_{k^k},k^k)) \xrightarrow {(t_1, \dots, t_k)} ((q'_1,1),\dots,(q'_{k^k},k^k))$$ 
					when  $(q_i \xrightarrow  {(t_{\varsigma_i(1)}, \dots, t_{\varsigma_i(k)})}q'_i )\in \delta$ for all $1\leq i \leq k^k$. For each $q_i$ the successor state $q'_i$ for a letter $(t_1,\dots,t_{k})$ is determined by the the transition of its permutation $(t_{\varsigma_i(1)}, \dots, t_{\varsigma_i(k)})$ in the automaton $\A$.  
			\end{itemize} 

The automaton $\A_{\mathfrak{P}}$ is a deterministic, permutation-complete and vertically tight. If the automaton $\A_\mathfrak{P}$ is not horizontally tight, it can then be translated to on  by redefining the set $F_\mathfrak{P}$ using the construction in Theorem~\ref{theo:constrtight}. The size of $A_\mathfrak{P}$ is $|Q|^{k^k}$.
\end{proof}

\begin{corollary}
	For a nondeterministic $k$-bad-prefix automaton $\A$ of some $k$-safety hyperproperty $\safety$ over $\Sigma$, we can construct a permutation-complete and tight $k$-bad-prefix automaton  of size $2^{|A|\cdot k^{k}}$.
\end{corollary}
The exponential blow-up in the size of the automaton in the last corollary results from the translation of nondeterministic automata to deterministic automata. 

\remark{Notice that the complexity in $k$ is independent of the representation of the $k$-safety-hyperproperty. 

%

\subsection{Equivalence of $k$-bad-prefix automata}
From the last section we know that the language of bad-prefixes for a safety hyperproperty is superset-closed. Thus every $k$-safety hyperproperty is also a $k'$-safety hyperproperty for all $k \leq k'$.   This means that $\safety$ can be represented by different bad-prefix automata of different arities $k'$. In the following we show, given a  $k'$-bad-prefix automaton and a $k''$-bad-prefix automaton with $k' \leq k''$, how to check whether they are bad-prefix automata for the same $k$-safety hyperproperty~$\safety$.


\begin{definition}[Representation-equivalence of bad-prefix automata]
 Let $A_{k'}$ be a finite automaton  over $\Sigma^{k'} $, and $A_{k''}$ be  a finite automaton  over $\Sigma^{k''}$ for some alphabet $\Sigma$, where $k' \leq k''$. 
	We say that $A_{k'}$ and $A_{k''}$ are \emph{representation-equivalent}, denoted by $A_{k'} \equiv A_{k''}$ if and only if both $A_{k'}$ and $A_{k''}$ are bad-prefix automata for the same $k$-safety hyperproperty $\safety$ for some $k \leq k',k''$. 
\end{definition}

An algorithm for checking equivalence of bad-prefix automata is given in the next theorem.

\begin{theorem}
	Let $\A_k$ be a deterministic finite automata  over $\Sigma^k$, and $\A_{k'}$ a deterministic finite automaton  over $\Sigma^{k'}$ for an alphabet $\Sigma$ and $k,k' \in \nats$. Checking whether $\A_k \equiv \A_{k'}$  can be done in time $O(\mathit{poly}(|\A_k|+|\A_{k'}|))$ and in space ${2^{O(\max(\log(k),\log(k'))\cdot\max\{k,k'\})}}$.  
\label{theo:equiv}
\end{theorem}
\begin{proof}
To check whether $A_k \equiv A_{k'}$ we have to check that:
\begin{enumerate}
 	\item For every representation $t$ accepted by $A_k$ and for every infinite extension $\widetilde{t}$ of $t$, there is a representation $t'$  accepted by $\A_{k'}$, such that, $t' \leq \widetilde{t}$: 
\vskip 0.1cm
$\forall T \in \unzip(L(A_k)).~\forall \widetilde{T}\subseteq \Sigma^\omega.\\~ T \leq \widetilde{T} \rightarrow  ~\exists T' \in \unzip(L(A_{k'})).~T' \leq \widetilde{T}$\\
 	\item For every representation $t'$ accepted by $A_{k'}$ and for every infinite extension $\widetilde{t}$ of $t'$, there is a representation $t$ accepted by $A_k$, such that, $t \leq \widetilde{t}$:
\vskip 0.1cm
$\forall T' \in \unzip(L(A_{k'})).~ \forall \widetilde{T} \subseteq \Sigma^\omega.\\~T' \leq \widetilde{T} \rightarrow  .~\exists T \in \unzip(L(A_k)).~T \leq \widetilde{T}$
 \end{enumerate}
\vskip 0.3cm
W.l.o.g. assume that $k<k'$. 
Let $\A_k = (Q_k, \Sigma^k,q_{0,k}, \delta_k, F_k)$ and  $\A_{k'} = (Q_{k'}, \Sigma^{k'},q_{0,k'}, \delta_{k'}, F_{k'})$. 
\begin{enumerate}
	\item To check the first direction, we first transform $\A_{k'}$ to a tight and permutation-complete automaton $\A_{k'}^{\mathfrak P}$ using the construction in Theorem~\ref{theo:constrpermcomplete}. To be able to compare $\A_k$ with $\A_{k'}^{\mathfrak P}$ we first expand the alphabet of $\A_k$ to $\Sigma^{k'}$ by constructing an automaton $A_k^{\uparrow k'}$ that preserves the language of $\A_k$ up to $\unzip(A_k)$. The automaton $A_k^{\uparrow k'}$ is defined by the tuple $(Q_k,\Sigma^{k'},q_{0,k}, \delta^{\uparrow k'}_k, F_k)$, where  $\delta^{\uparrow k'}_k(q, {(t_1,\dots, t_k, \dots, t_k')}) = q'$ if and only if $t_i =t_k$ for all $k <i \leq k'$ and $\delta_k(q,(t_1, \dots, t_k)) = q'$, otherwise there is no transition. Clearly, $\unzip(A_k) = \unzip(A_k^{\uparrow k'})$. 

	We build the product automaton $\A_{\otimes}$ of $A^{\uparrow k'}_k$ and $\A^{\mathfrak P}_{k'}$. If $\A_{\otimes}$ has a lasso run\footnote{This is an infinite run in the automaton that can be represented by a sequence of states that reach a loop in the automaton.}, where there is an accepting state of $A^{\uparrow k'}_k$ but no accepting states of $A_{k'}^{\mathfrak P}$, then condition~(1) is violated and thus $A_k \not \equiv \A_{k'}$. If no such run is found, then $A_k \equiv A_{k'}$. 
	    
	    The size of the product automaton is $|Q_k| \cdot |Q_{k'}|^{k'^{k'}} $. To check the equivalence there is no need to construct the automaton $\A_{\otimes}$ in fully. Using the same trick as in the polynomial-space model checking algorithm for LTL~\cite{Baier:2008:PMC:1373322}, we can guess a lasso run in $\A_{\otimes}$ of size at most $|\A_\otimes|$. The lasso can be guessed one position at a time and in each position on can further guess if it is the beginning of the period of the lasso. In each step we check if the guessed next position of the lasso satisfies the transition relation as given in the construction above.  Finding a lasso in $\A_{\otimes}$ can thus be done in time polynomial in the sizes of $A_k$ and $A_{k'} $  and in space exponential in $k'$. 

	\item For the other direction it does not suffice to construct the tight and permutation-complete automaton for $\A_k$ and check the condition (2) on the product automaton with $A_{k'}$ as we did in the last case. The reason why this construction does not work, is due to the different arities $k$ and $k'$. Intuitively, $A_k$ and $A_{k'}$ are equivalent, if for each representation accepted by $A_{k'}$, a permutation of one of its $k$-projections satisfies the condition (2). To this aim we construct an automaton $\A_k^{\# k'} = (Q^{\# k'},\Sigma^{k'},q^{\# k'}_{0}, \delta^{\# k'}, F^{\# k'})$ as follows:
	
	Let  $\varsigma_1, \dots, \varsigma_{k'^k}: \{1, \dots, k\} \rightarrow \{1, \dots, k'\}$ be pairwise different functions. We call $\varsigma_1, \dots, \varsigma_{k'^k}$ $k$-permuted-projection functions.
	\begin{itemize}
		\item $Q^{\# k'} = (Q_{k,1}\times \dots \times Q_{k,k'^k})$ where $Q_{k,i} = \{(q,i) \mid q\in Q_k\}$ for $1\leq i \leq k'^k$. A set of states $Q_{k,i}$ resembles a copy of the automaton $\A$ that accepts a word $\sigma$ if one of its  $k$-permutated-projections is accepted by $\A$ with respect to the permuted-projection function $\varsigma_i$.  The initial state is defined by $q^{\# k'}_{0} = ((q_0,1),\dots,(q_{0},k'^k))$.
		\item A word is accepted, if one of its permuted-projections is accepted. We define the set of accepting states as $F_{\mathfrak P} = \{((q_1,1), \dots, (q_{k'^k},k'^k)) \mid \exists i.~q_i \in F \}$.
		\item The transition relation $\delta^{\# k'}_{k}$ is defined as follows:
		
					$((q_1,1),\dots,(q_{k'^k},k'^k))\\ ~~~~~~~~~~~~~~~~ \xrightarrow {(t_1, \dots, t_{k'})} ((q'_1,1),\dots,(q'_{k^k},k^k))$\\
					when  $(q_i \xrightarrow  {(t_{\varsigma_i(1)}, \dots, t_{\varsigma_i(k)})}q'_i )\in \delta$ for  $1\leq i \leq k'^k$. For each $q_i$ the successor state $q'_i$ for a letter $(t_1,\dots,t_{k'})$ is determined by  the transition of its permuted-projection $(t_{\varsigma_i(1)}, \dots, t_{\varsigma_i(k)})$ in the automaton $\A$. 
	\end{itemize} 
	We build the product automaton $\A_{\otimes}$ of $A^{\# k'}_k$ and $\A_{k'}$. If $\A_{\otimes}$ has a lasso run, where there is an accepting state of $A_{k'}$ but no accepting states of $A^{\# k'}_k$, then condition~(2) is violated and thus $A_k \not \equiv \A_{k'}$. If no such run is found, then $A_k \equiv A_{k'}$. 
	 
	 Again, To check the equivalence we can guess a lasso in $\A_{\otimes}$ that has an accepting state of $\A_{k'}$ but no accepting state from $A^{\# k'}_k$. Finding a lasso in $\A_{\otimes}$ can thus be done in time polynomial in the sizes of $A_k$ and $A_{k'} $  and in space exponential in $k$. 
\end{enumerate}	
\end{proof}

\begin{corollary}
	 Checking the representation-equivalence of two non-deterministic finite automata $A_k$ and $A_{k'}$ of arity $k,k'$ can be done in space $O(\mathit{poly}(|A_k|+|A_{k'}|))$ and ${2^{O(\max(\log(k),\log(k'))\cdot\max\{k,k'\})}}$.
\end{corollary}
%
%
\subsection{Minimal $k$-bad-prefix automata}
In this section, we complete our search for a canonical representation of regular $k$-safety hyperproperties and prove that minimal deterministic, tight and permutation-complete {bad-prefix} automata provide such a representation. 

\begin{definition}[Minimal Bad-prefix Automaton]
A deterministic tight permutation-complete $k$-bad-prefix automaton $\A$ for some $k'$-safety hyperproperty $\safety$ is called \emph{minimal}, if there is no $k''$-bad-prefix automaton for $\safety$, with $k''<k$, and $\A$ is the minimal automaton in size for $k$.
\end{definition}

\begin{lemma}
	\label{lem:equivalence_and_bad-prefixes}
	Two safety hyperproperties $\safety$ and $\safety'$ are equivalent if and only if $\bad(\safety) = \bad(\safety')$. 
\end{lemma}%
\begin{proof}
	\begin{itemize}
		\item[($\Rightarrow$)] Let $T \in \bad(\safety)$. Thus for all $T \leq T'$, it follows, that $T' \not \models \safety$ and therefore $T' \not \models \safety'$, by assumption.
		Hence, $T \in \bad(\safety')$. The same proof holds when exchanging $\safety$ and $\safety'$ yielding $\bad(\safety) =  \bad(\safety')$.
		\item[($\Leftarrow$)] Let $T \not \in \safety$. Thus there exists some $T' \in \bad(\safety)$ such that $T' \leq T$. According to the assumption we know, that $T' \in \bad(\safety')$ and thus $T \not \in \safety'$. 
		The same proof holds when exchanging $\safety$ and $\safety'$ yielding $\safety =  \safety'$.
	\end{itemize}
\end{proof}
\begin{theorem}
	Minimal, tight, permutation-complete, deterministic $k$-bad-prefix automata are a canonical representation for regular $k$-safety hyperproperties.
\end{theorem}%
\begin{proof} 
	Let $\safety$ and $\safety'$ be two regular $k$-safety hyperproperties. We  show that they are equal if and only if they have the same minimal deterministic tight permutation-complete bad-prefix automaton. 
		
	\begin{itemize}
	\item [$(\Rightarrow)$]
	Let $\safety = \safety'$. From Lemma~\ref{lem:equivalence_and_bad-prefixes}
	 we know that $\bad(\safety) = \bad(\safety')$. This means that any representation $L$ for $\safety$ is also a representation for $\safety'$. We conclude that any deterministic tight permutation-complete bad-prefix automaton for $\safety$ is also a deterministic tight permutation-complete bad-prefix automaton for $\safety'$.
	\item[$(\Leftarrow)$]
	Let $\A$ be a  $k$-bad-prefix automaton for $\safety$ and $\safety'$. This means  
	that $\bad(\safety) = \bad(\safety')$. From \Cref{lem:equivalence_and_bad-prefixes} it follows that $\safety \equiv \safety'$.
	\end{itemize}
	
	It remains to show that  $k$-bad-prefix automata are unique for a $k$-safety hyperproperty $\safety$. 
	Clearly, the minimal arity $k$ is unique. 
	The language of all bad-prefixes of size $k$ is also unique for $\safety$. 
	Thus the set of all $k$-representations is unique and  this language is regular by assumption.
	Further, from the Myhill-Nerode Theorem. it is well-known that minimal deterministic automata are a unique representation for regular languages\cite{hopcroft:introduction_to_automata_theory}
	Hence, the claimed uniqueness follows. 
\end{proof}%

Based on this canonical representation, we provide, in the next section, a framework for learning automata for regular $k$-safety hyperproperties.

\section{Learning Automata for  $k$-Safety Hyperproperties}\label{sec:Learning}
\label{sec:learning}
We present a framework for learning \emph{minimal tight, permutation-complete,  deterministic} bad-prefix automata  for some unknown $k$-safety hyperproperty $\safety$ over an alphabet $\Sigma$ and an unknown minimal $k$. The algorithm extends Dana Angluin's L$^*$ algorithm for learning minimal deterministic finite automata from queries and counterexamples \cite{angluin_learning_regular_sets}, to learn minimal bad-prefix automata  for a minimal arity $k$. 

\subsection{L$^*$: A framework for learning regular languages}
We give a high-level recap of the L$^*$ framework as presented in Figure~\ref{fig:L*framework}. We leave some of the technical details for the next section when explaining the extended framework for $k$-safety-hyperproperties. 

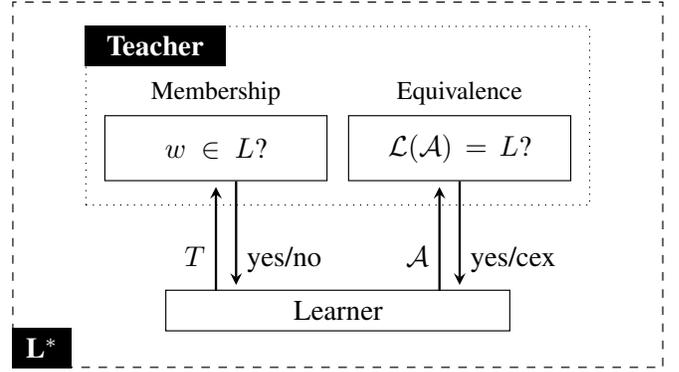
\begin{figure}[h]
	\centering
	{
		\scalebox{1.08}{\centering
\begin{tikzpicture}

\draw  [dashed](-4,-0.5) rectangle ++ (8,4.5);

\draw  [dotted](-3.1,1.5) rectangle ++ (6.2,2.2);

\node [draw = black, text width=4cm, minimum height = 0.5cm, align=center] (learner) at (0,0.2) {Learner};

\node [draw = black, text width=2.5cm, minimum height=0.8cm,  align=center] (MQ) at (-1.5,2.2) {$w \in L$?};

\node [draw = black, text width=2.5cm, minimum height=0.8cm, align=center] (EQ) at (1.5,2.2) {$ \L(\A) = L$?};

\node [fill = black, text width=.5cm, minimum height=0.5cm, align=center](Lstar) at (-3.64,-0.26) {\color{white}\textbf{L}$^*$};

\node [fill = black, text width=1.5cm, minimum height=0.5cm, align=center](teacher) at (-2.25,3.46) {\color{white}\textbf{Teacher}};

\node [text width=2.5cm, minimum height=0.5cm, align=center, above = 0 of MQ](teacher) {\small Membership};

\node [text width=2.5cm, minimum height=0.5cm, align=center, above = 0 of EQ](teacher) {\small Equivalence};

\node [below = 0 of MQ] (r1) {};
\node [below = 1.1 of r1] (r2) {};
\node [right = 0 of r2] (r3) {};
\node [right = 0 of r1] (r4) {};

\node [below = 0 of EQ] (l1) {};
\node [below = 1.1 of l1] (l2) {};
\node [left = 0 of l2] (l3) {};
\node [left = 0 of l1] (l4) {};

\path[->]
(r2) edge node [label, below left] {$T$} (r1.north)
(r4.north) edge node [label, below right] {yes/no} (r3)
(l1.north) edge node [label, below right] {yes/cex} (l2)
(l3) edge node [label, below left] {$\A$} (l4.north)
;

\end{tikzpicture}}
	}
	\caption{L$^*$: A framework for learning minimal deterministic finite automata~\cite{angluin_learning_regular_sets}.}
	\label{fig:L*framework}
\end{figure}

The L$^*$ framework consists of two components, a \emph{learner}, that learns an automaton for the unknown language, and a \emph{teacher}, that answers questions about the language. The learner can pose two types of queries to the teacher: \emph{membership-queries}, where the learner asks whether a word is in the target language, and \emph{equivalence-queries}, where the learner asks whether the language of a conjectured automaton is equivalent to the target language. 

The learner starts by posing membership queries for words of increasing length.
The answers of the teacher are organized in a so-called \emph{observation table}.  The observation table represents a so-far constructed automaton, the accepts at least the valid words queried using membership queries. 
 After each membership query, the learner performs two checks in the observation table: (1) a consistency check, certifying that the observation table defines a deterministic automaton; the table contains no two transitions from a state for the same letter, and (2) a closedness check, that tests that it defines a complete automaton, i.e., for each state and for each letter there is a transition from that state for this letter.  If one of these checks fails, the observation table can be repaired with the appropriate extension and membership queries (We show how these checks are performed in the case of $k$-safety hyperproperties in the next section. For the traditional checks for regular languages we refer the reader to \cite{angluin_learning_regular_sets}).   
 
 If the table is both consistent and closed, then the learner can construct a deterministic automaton out of the observation table and queries the teacher on whether the conjectured automaton defines the target language. If the automaton is not equivalent to the target language, the teacher returns a counterexample. This is either a word in the language that is not accepted by the conjectured automaton, or a word that is wrongly accepted by the automaton and is not a member of the target language. The counterexample is added to the observation table, and the learning process continues with the new table.  

Angluin showed that, for a \emph{minimal adequate teacher}, i.e., a teacher that answers membership and equivalence queries,  that L$^*$ terminates after a number of membership queries that is polynomial in the size of the minimal deterministic finite automaton for the target language.  

\subsection{L$^*_\mathit{Hyper}$: A framework for learning $k$-safety hyperproperties}

\begin{figure*}
\centering
{
	\scalebox{.8}{\begin{tikzpicture}

\def \height {2cm}
\node (start) {};

\node [draw=black, text width=3cm, minimum height=\height, rounded corners=0.5cm, right=3cm of start, align=center] (close_and_consistent) {is table closed and consistent?};

\node [draw=black, text width=3cm, minimum height=\height, rounded corners=0.5cm, right=2cm of close_and_consistent, align=center] (conjecture) {build conjecture bad-prefix $\A$};

\node [draw=black, text width=3cm, minimum height=\height, rounded corners=0.5cm, right=2cm of conjecture, align=center] (equiv) {is $\A$ equivalent?};

\node [draw=black, text width=3cm, minimum height=\height, rounded corners=0.5cm, above=2cm of equiv, align=center] (permutation) {is $\A$ permutation-complete?};

\node [right=2cm of permutation, align=center] (acc) {};

\node [draw=black, text width=3cm, minimum height=\height, rounded corners=0.5cm, below=2cm of equiv, align=center, align=center] (arity) {arity of $C$ less or equal $k$};

\node [draw=black, text width=3cm, minimum height=\height, rounded corners=0.5cm, below=2cm of close_and_consistent, align=center] (extend) {extend $\mc{O}$ to $\Sigma^{\abs{C}}$};

\path[->]
(start) edge node [label, above] {inital table} (close_and_consistent)
(close_and_consistent) edge [loop above] node [label, above] {no: membership queries} (close_and_consistent)
(close_and_consistent) edge node [label, above]{yes} (conjecture)
(conjecture) edge node [label, above] {$\A$} (equiv)
(equiv) edge node [label, right] {no: counterexample $C$} (arity)
(equiv) edge node [label, right] {yes} (permutation)
(permutation) edge node [label, above =1] {no: add counterexample $C$ to $\mc{O}$} (close_and_consistent)
(permutation) edge node [label, above] {yes: done} (acc)
(arity) edge node [label, above] {no: counterexample $C$} (extend) 
(arity) edge node [label, above right] {yes: add $C$ to $\mc{O}$} (close_and_consistent)
(extend) edge node [label,left]{add $C$ to $\mc{O}$}(close_and_consistent);

\end{tikzpicture}}
}
\caption{\lstarhyper: A framework for learning $k$-safety-hyperproperties. }
\label{fig:Lstarhyper}	
\vskip 2mm
\end{figure*}
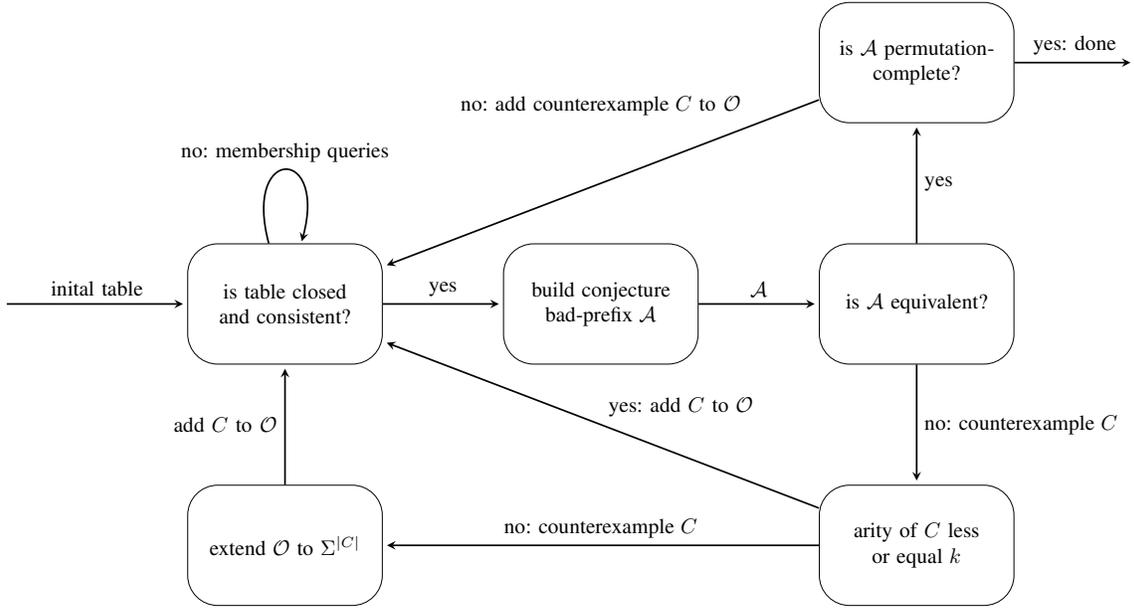

We extend the L$^*$ framework given in Figure~\ref{fig:L*framework} to a new framework \lstarhyper\ for learning minimal deterministic tight permutation-complete bad-prefix automaton for $k$-safety-hyperproperties. In contrast to learning minimal automata for regular properties, the bad-prefix automata learned in \lstarhyper\ must be minimal both in the arity and in the size. The work-flow of \lstarhyper is given in  Figure~\ref{fig:Lstarhyper}.

Let $\safety$ be the unknown $k$-safety-hyperproperty over some alphabet $\Sigma$. The learner starts with membership queries, and fills the observation table until it is closed and consistent. Because the minimal arity $k$ is initially unknown the learner starts by posing questions over sets of arity 1. During the learning process the alphabet  changes to larger arity ${k'}$, when the teacher returns a counterexample of this arity.
 
Assume the current arity in the learning process is $k'$ for some $k'\leq k$. In membership queries, the learner asks  whether a finite set of finite sequences of equal length $T = \{t_1,\dots,t_{k'}\} \subseteq \Sigma^n$ (given by some representation) for some $n\in\nats$ is a bad-prefix for $\safety$. 
The answers of the teacher are organized in an observation table $\O = (S,E,\Delta)$ 
where: $S \subseteq (\Sigma^{k'})^*$ is a non-empty finite prefix-closed set of \emph{accessing sequences},  $E\subseteq (\Sigma^{k'})^*$ is a non-empty finite suffix-closed set of \emph{separating sequences}, and $\Delta: (S \cup S\cdot \Sigma) \cdot E \rightarrow \{0,1\}$ a mapping defined as $\Delta(s\cdot e) = 1$ if and only if $s\cdot e$ is a representation of a bad-prefix for $\safety$. 

Consider the observation table given in Figure~\ref{fig:obsertableArity1}. The set $S$ includes the words $\epsilon,\neg a, a, a\cdot\neg a$ in the first four rows of the table.  The set $E$ includes the words $\epsilon,\neg a$ defining the columns. The words in the remaining columns define the set $S\cdot \Sigma$ (as we will see later, these rows are necessary for the closedness and consistency checks).  The value of an entry in the table is 1 if the word $s\cdot e$, where $s$ is a word of the row and $e$ the word of the column, is a representation of a bad-prefix for the hyperproperty given by the formula $\forall \pi, \pi'.\ a_{\pi} \wedge \LTLsquare (a_{\pi} \leftrightarrow a_{\pi'})$. Otherwise the value of the entry is 0.  

For $t \in S\cdot\Sigma$ we denote by $\row(t)$ a finite function from $E$ to $\{0,1\}$ defined by $\row(t)(e) = \Delta(t\cdot e)$. 
An observation table $\O = (S,E,T)$ is called closed if for all $t \in S \cdot \Sigma$ there exists $s \in S$ with $\row(t) = \row(s)$. The table 
 $\O$ is called consistent, if for all $t, t' \in S$ with an equal function $\row(t) = \row(t') \Rightarrow \row(t\cdot e) = \row(t' \cdot e)$ for all $e \in \Sigma$.  We define $\row(S) = \{\row(s) \mid s \in S\}$. Closedness guarantees that every transition is defined, i.e., for each state $q \in Q$ and label $a \in \Sigma$, $\delta(q,a) \in Q$, and consistency guarantees that $\A$ is deterministic. 

For a closed and consistent observation table $\mc{O}$ over an arity $k'$ we can construct an DFA $\A = (Q, \Sigma^{k'}, q_0, F, \delta)$ that accepts all the $k'$-bad-prefixes that have been confirmed by the teacher so far. We define $Q = \{ \row(s) \mid s \in S\}$, $q_0 =\row(\epsilon)$, $F = \{\row(s) \mid s\in S \text{ and } \Delta(s) = 1\}$ and $\delta(\row(s), a) = \row(s\cdot a) $ for all $s \in S$ and $a \in \Sigma$.
The table in Figure~\ref{fig:obsertableArity1} is closed and consistent, and defines the automaton given to its right. 
For $k' \in \NN$, we call an automaton $\A$ over $\Sigma^k$ consistent with an observation table $\O = (S,E,\Delta)$ over $\Sigma^k$ if for all $s\in (S\cup S \cdot \Sigma^k), e \in E$ $\Delta(s, e) = 1 \Leftrightarrow s\cdot e \in L(\A)$.



To check whether the learned automaton $\A$ is a $k'$-bad-prefix automaton for $\safety$,  the learner poses an equivalence query to the teacher. 
In equivalence queries, the teacher answers whether the proposed automaton $\A = (Q, \Sigma^{k'}, q_0, F, \delta)$ is a $k'$-bad-prefix automaton for $\safety$. In case $\A$ is not, the teacher provides a counterexample. 

\begin{figure}[h]
	\begin{minipage}{0.49\linewidth}
		\centering
		\scalebox{0.9}{
		\begin{tikzpicture}
			\matrix at(0,0) (m) [mat, column 1/.style={nodes={minimum width = 7em}},ampersand replacement=\&]
			{							\& $\epsilon$ 	\&  $\neg a$ \\
				$\epsilon$ 				\& 	$0$ 		\&	$1$		\\
				$\epsilon\cdot\neg a$ 	\& 	$1$ 		\& 	$1$		\\
				$\epsilon \cdot a$ 		\&	$0$ 		\& 	$0$		\\
				$a \cdot \neg a$		\&  $0$			\& 	$0$		\\
				$\neg a\cdot a$ 		\& 	$1$ 		\& 	$1$ 	\\
				$\neg a\cdot \neg a$ 	\& 	$1$ 		\&	$1$		\\
				$a \cdot a$				\&  $0$ 		\& 	$0$		\\
				$a \cdot \neg a \cdot a$		\& 	$0$	\&	$0$		\\
				$a \cdot \neg a \cdot \neg a$	\& 	$0$	\& 	$0$		\\ 
			};
			\draw[line width = 0.9 mm] (m-6-1.north west) -- (m-6-3.north east) ;
		\end{tikzpicture}
		}
	\end{minipage}
	\begin{minipage}{0.49\linewidth}
		\centering
		\scalebox{0.9}{
		\begin{tikzpicture}
			\node [comp, initial] 	(q0) {$q_0$};
			\node [comp, accepting, right = 5em of q0]	(q1) {$q_1$};
			\coordinate (Middle) at ($(q0)!0.5!(q1)$);
			\node [comp, above = 5em of Middle] (q2) {$q_2$}; 
			
			\path[->]
			(q0) edge 				node [label, above] {$\neg a$} 	(q1)
			(q0) edge 				node [label, above left] {$a$} 		(q2)
			(q1) edge [loop above] 	node [label, above] {$\top$} 		(q1)
			(q2) edge [loop above] 	node [label, above]	{$\top$}		(q2)
			;
		\end{tikzpicture}
		}
	\end{minipage}
\caption{The  observation table (on the left) for the $4^\text{th}$ iteration in the learning process for the language $\forall \pi, \pi'.\ a_{\pi} \wedge \globally (a_{\pi} \leftrightarrow a_{\pi'})$, and the corresponding DFA (on the right).}
\label{fig:obsertableArity1}
\end{figure}
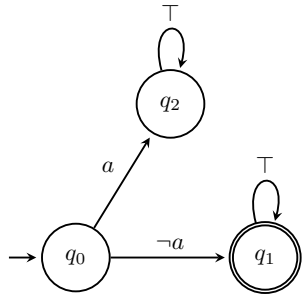

\subsection{Handling counterexamples of equivalence queries}
If the equivalence test fails, then the teacher returns a counterexample. A counterexample is either a bad prefix for which no representation is accepted by the conjectured automaton, or a representation accepted by the automaton that is no representation of a bad prefix for $\safety$. 

%

Handling bad prefixes depends on their arity. 
We distinguish between two types of counterexamples with respect to the current considered arity $k'$, namely, counterexamples with arity $k'' \leq k'$ and counterexamples with arity $k''> k'$.

 If the counterexample has arity $k''\leq k$, the counterexample is treated as for the traditional L$^*$ learner, by extending the table with a representation of this counterexample and querying all its prefixes. 

 If the counterexample has arity $k''>k'$  then the arity of the target automaton is increased to $k''$. The sets $S$ and $E$ are extended to sequences over the alphabet $\Sigma^{k''}$ by replacing every element $t = v_0\dots v_n \in (\Sigma^{k'})^*$ in $S$ and $E$ by $t' = v_0'\dots v_n' \in (\Sigma^{k''})^*$, such that, for all $0\le i\le n$ and $0 \le j \le k'$, $v_{i}[j] = v_{i}'[j]$ and for $ k' < j \le k''$, $v_{i}[k'] = v_{i}'[j]$. Notice that the size of the table increases by the number of prefixes of the counterexample. 
 
 Consider again our example in Figure~\ref{fig:obsertableArity1}. The conjecture automaton is not a bad-prefix automaton for the language $\forall \pi, \pi'.\ a_{\pi} \wedge \LTLsquare (a_{\pi} \leftrightarrow a_{\pi'})$. A counterexamples of arity 2 is given by the set $C = \{\{a \cdot \neg a\},\{ a \cdot a\}\}$. The observation table must be extended to $\Sigma^2$. Having extended the observation table, we must add a representation of the counterexample $C$ to the set of accessing sequences $S$. $C$ is a bad prefix, which can be verified using a membership query. As no 2-representation of $C$ are in the target language.  We choose the 2-representation $(a,  a) · (\neg a, a)$ and add it to the set of accessing sequences resolves the current counterexample. The new observation table  is closed and consistent, and it represents a 2-bad-prefix automaton that represents the target language. The final table and its automaton are depicted in Figure~\ref{fig:obsertableArity2}.

\begin{figure}	
\begin{minipage}{1\linewidth}
		\centering
		\scalebox{0.9}{
		
		\begin{tikzpicture}
		\matrix at(0,0) (m) [mat, column 1/.style={nodes={minimum width = 12em}}, nodes = {minimum width = 5em}, ampersand replacement=\&]
		{												\& $\epsilon$ 	\&  $(\neg a, \neg a)$ \\
			$\epsilon$ 									\& 	$0$ 		\&	$1$		\\
			$\epsilon\cdot(\neg a, \neg a)$ 			\& 	$1$ 		\& 	$1$		\\
			$\epsilon \cdot (a,a)$ 						\&	$0$ 		\& 	$0$		\\
			$(a,a) \cdot (\neg a, \neg a)$				\&  $0$			\& 	$0$		\\
			$(a,a) \cdot (\neg a, a)$					\&  $1$			\& 	$1$		\\
			$(a, \neg a)$								\& 	$1$			\& 	$1$		\\
			$(\neg a, a)$ 								\& 	$1$			\& 	$1$		\\
			$(\neg a, \neg a)\cdot (*)$ 				\& 	$1$ 		\& 	$1$ 	\\
			$(a,a) \cdot (a, a)$						\&  $0$ 		\& 	$0$		\\
			$(a, a) \cdot (\neg a, a)$					\&  $1$			\& 	$1$		\\
			$(a,a) \cdot (\neg a, \neg a) \cdot (a,a)$	\& 	$0$			\&	$0$		\\
			$(a,a) \cdot (\neg a, \neg a) \cdot (\neg a,a)$	\& 	$1$			\&	$1$		\\
			$(a,a) \cdot (\neg a, \neg a) \cdot (a,\neg a)$	\& 	$1$			\&	$1$		\\
			$(a,a) \cdot (\neg a, \neg a) \cdot (\neg a,\neg a)$	\& 	$0$			\&	$0$		\\
			$(a,a) \cdot (\neg a, a) \cdot (*)$					\&  $1$			\& 	$1$		\\
		};
		\draw[line width = 0.9 mm] (m-7-1.north west) -- (m-7-3.north east) ;
		\end{tikzpicture}
		}
	\end{minipage}
	~~~~~~~~~~~~{\color{white}.} 
	\vskip 5mm
	\begin{minipage}{1\linewidth}
	\centering
		\scalebox{1}{
		
		\begin{tikzpicture}
		\node [comp, initial] 	(q0) {$q_0$};
		\node [comp, accepting, right = 5em of q0]	(q1) {$q_1$};
		\coordinate (Middle) at ($(q0)!0.5!(q1)$);
		\node [comp, above = 5em of Middle] (q2) {$q_2$}; 
		
		\path[->]
		(q0) edge 				node [label, above] {$\neg (a, a)$} 	(q1)
		(q0) edge 				node [label,  left] {$(a,a)$} 		(q2)
		(q1) edge [loop right] 	node [label, right] {$\top$} 		(q1)
		(q2) edge [loop above] 	node [label, above]	{$(a, a) \vee (\neg a, \neg a)$}		(q2)
		(q2) edge 				node [label, right]		{$(\neg a, a) \vee (a, \neg a)$}		(q1)
		;
		\end{tikzpicture}
		}
	\end{minipage}
	\caption{Final observation table for learning $\varphi = \forall \pi, \pi'.\ a_{\pi} \wedge \globally (a_{\pi} \leftrightarrow a_{\pi'})$, and $2$-bad-prefix automaton for $\varphi$.  We employ the notation $ x\cdot(*)$  to denote  all extensions of $x$.}
 \label{fig:obsertableArity2}	
 \end{figure}
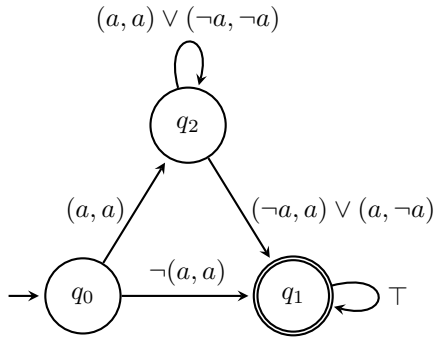
 
 In some case, the resulting automaton might pass the equivalence check but is not necessarily permutation complete. If the goal is to construct a permutation-complete automaton, we additionally check the automaton for permutation-completeness.


\subsection{Termination of \lstarhyper}

As for \lstar, to guarantee that our learning framework terminates with a minimal automaton, it must rely on a minimal-adequate teacher. For \lstarhyper we define such a teacher as follows. 

\begin{definition}[Minimal-Adequate Teacher]
\label{def:min_ade_teacher}
	A teacher is called minimal-adequate, if the counterexamples provided are of minimal length and every counterexample has an arity of at most the minimal target-arity. 
\end{definition}

For a minimal adequate teacher we show that for an unknown $k$-safety hyperproperty $\safety$, \lstarhyper terminates after at most $n$ equivalence queries, where $n$ is the size of the minimal deterministic tight permutation-complete $k$-bad-prefix automaton for $\safety$. 

Let $\O=(S,E,\Delta)$ be the observation table. Following the ideas in the termination proof for \lstar, we need to show that 
\begin{enumerate}
	\item the set $\row(S)$ does not grow beyond $n$.
	\item with each counterexample, the size of the set $\row(S)$ must be strictly monotonically increasing.
\end{enumerate}
In the following lemma we give a proof for condition 1). With Lemmas~\ref{lemma:consistent}, \ref{lemma:closedness}, \ref{lemma:addCEX} and \ref{lemma:extend}, we show that \lstarhyper also satisfies  condition 2).\\ 

 \begin{lemma}
 \label{lemma:boundedness}
 	Let $\safety$ be a regular $k$-safety hyperproperty, let $\O=(S,E,\Delta)$ be an observation table, and let $n$ be the size of the minimal $k$-bad-prefix automaton for $\safety$. Then, the size of $\row(S)$ is bounded by $n$ in any iteration of \lstarhyper.  
 \end{lemma}
 \begin{proof}
 	 The proof follows from the fact that, for any closed and consistent observation table, an automaton that is consistent with this table has at least $n$ states \cite{angluin_learning_regular_sets}. 
 \end{proof}

To prove condition 2), we need to show that after every counterexample, the set $\row(S)$ must increase by at least one. This requires us to show that 
\begin{itemize}
	\item resolving consistency and closedness of observation tables increases the size of $\row(S)$ by at least one. (Lemma~\ref{lemma:consistent} and Lemma~\ref{lemma:closedness})
	\item adding a counterexample makes the observation table $\cal O$ inconsistent or not closed (Lemma~\ref{lemma:addCEX})
	\item extending the observation table from an arity $k'$ to \\ $k''>k'$ preserves the size of $\row(S)$ (Lemma~\ref{lemma:extend})
\end{itemize}

We define a witness of inconsistency in $\O$ as a triple $(s, t, a) \in S \times S\times \Sigma$ such that $\row(s) = \row(t)$ and $\row(s\cdot a) \neq \row(t\cdot a)$ and a witness of non-closedness in $\O$ as a word $w \in S\cdot \Sigma$ such that $\row(w) \neq \row(s)$ for all $s \in S$. 

The proofs for the following three lemmas are given in Dana Angluin's termination proof for $L^*$\cite{angluin_learning_regular_sets}.


\begin{lemma}
	\label[lemma]{lemma:consistent}
	Let $\mathcal{O}=(S,E, \Delta)$ be an inconsistent observation table over $\Sigma$ with a witness $(s, t, a)$. Resolving this witness increases the size of $\row(S)$.  
\end{lemma}


\begin{lemma}
\label[lemma]{lemma:closedness}
Let $\mathcal{O}=(S,E, \Delta)$ be an observation table over $\Sigma$ that is not closed with a witness $w \in S\cdot \Sigma$. Resolving this witness increases the size of $\row(S)$.  
\end{lemma}

The next lemma states that amending an observation table with a counterexample results in an inconsistent observation table.

\begin{lemma}
\label[lemma]{lemma:addCEX}
  Let $\mathcal O = (S,E,\Delta)$ be a consistent observation table over an alphabet $\Sigma$. Let $w \in \Sigma^*$, be a counterexample resulting from an equivalence check or from a check of permutation-completeness. 
 Then the resulting observation table $\mathcal{O}' = (S',E',\Delta')$ that results from $\mc{O}$ by amending it with $w$ is either inconsistent or not closed. 
 
\end{lemma}

Obtaining a counterexample containing $k'$ many traces, which can not be represented by the learner's current arity $k_L$, i.e., $k' > k_L$, the observation table needs to be extended to $k'$ before adding the counterexample. We want to achieve this extension without losing information obtained by earlier queries, i.e., without changing the size of $\row(S)$. 
Therefore, we extend each representation in $S\cup E$ to a representation in $\Sigma^{k'}$ by repeating its last position. After this extension every representation still represents the same set of traces, since we only repeat one of the traces in trace set. The procedure of extending an observation table from arity $k_L$ to a larger arity $k$ is given in  \Cref{algo:extend}. The function $\extend$ is defined as follows: For a tuple $(t_1,\dots,t_{k_L}) \in \Sigma^{k_L}$, $\extend(t_1,\dots,t_{k_L},k')=(t_1,\dots,t_{k_L},\dots,t_{k'})$, where $t_i = t_{k_L}$ for all $k_L < i \leq k'$.

\begin{algorithm}[h]
	\caption{\textsc{Extend}}
	\label{algo:extend}
	\hspace*{\algorithmicindent} \textbf{Input} Observation Table $\O = (S,E,\Delta)$ over $\Sigma^{k_L}$, $k'> k_L$ \\
	\hspace*{\algorithmicindent} \textbf{Output} Observation table $\O'$ over $\Sigma^{k'}$ \vfill
	\begin{algorithmic}[1]
		\STATE $\O' = (S',E',\Delta') =  (\{\epsilon\}, \{\epsilon\}, \{((\epsilon, \epsilon), \Delta(\epsilon, \epsilon))\})$
		\FOR {$s \in S$}
			\STATE $S' = S' \cup \{\extend(s,k')\}$
		\ENDFOR
		\FOR{$e \in E$}
			\STATE $E' = E' \cup \{\extend(e,k')\}$
		\ENDFOR
		\STATE fill $\Delta'$ for $(S'\cup S'\cdot \Sigma^{k'})\cdot E'$ using Membership queries
		\RETURN $\O'$.
	\end{algorithmic}
\end{algorithm}

\begin{lemma}
\label[lemma]{lemma:extend}
Let $\mathcal O = (S,E,\Delta)$ be an observation table over an alphabet $\Sigma^{k'}$ for some $k'\in \nats$.
Let $\mathcal {O}' = (S',E',\Delta')$ be the observation table obtained by applying {\sc Extend} (Algorithm~\ref{algo:extend}) to $\mathcal O$ for some $k''>k'$. Then, $|\row(S)| = |\row(S')|$.
\end{lemma}
\begin{proof}
	Prove by contradiction. Assume $\abs{\row(S)} \neq \abs{\row(S')}$. We only treat one direction the other direction can be proven equivalently. Let $\abs{\row(S)} >\abs{\row(S')}$, i.e., there exist $s,t \in S$ and $e \in E$ such that the following (in-)equalities hold 
	\begin{align*}\Delta(s, e) &\neq \Delta (t, e)\\
	\Delta(\extend_{k''}(s), \extend_{k''}(e)) &= \\
	& \Delta(\extend_{k''}(t), \extend_{k''}(e)).
	\end{align*}
	W.l.o.g., let $\Delta(\extend_{k''}(t), \extend_{k''}(e))\neq \Delta (t, e)$, i.e., the result of a membership queries for $\extend_{k''}(t\cdot e)$ and $t\cdot e$ differ. Since $\extend_{k''}$ does not alter the represented set of words, such words $s, t$, and $e$ cannot exist. A contradiction to our assumption. 
\end{proof}

Using the lemmas above, we are now able to prove the termination of \lstarhyper.
\begin{theorem}
	\label{theo:progress}
	L$^*_{Hyper}$ terminates after at most $n$ equivalence queries, where $n$ is the size of the minimal deterministic tight permutation-complete $k$-bad-prefix automaton for $\safety$. 
\end{theorem}

\subsection{Complexity of \lstarhyper}

In the last section we showed that for a minimal-adequate teacher, \lstarhyper  terminates after at most $n$ equivalence queries, where $n$ is the size of the minimal bad-prefix automaton for the target $k$-safety hyperproperty $\safety$. Preceding every equivalence query, the learner checks the consistency and closedness of the observation table, and after the equivalence check the learner may need to check the permutation-completeness of the conjectured automaton or extend the observation table to a higher arity, and then amend the observation table with a new counterexample.

After each operation, the number of rows in the observation table is increased by at least one, but cannot increase beyond $n$, as we have shown in Lemma~\ref{lemma:boundedness}. This means that we can perform at most $n$ many of these operations. 

For $L^*$, Angluin showed that the learning algorithm is polynomial in $n$ and in the length of the largest returned counterexample. This complexity also holds for \lstarhyper. The runtime of \lstarhyper also depends on the goal arity $k$. The runtime complexity in $k$ can be derived by studying the runtime of the procedures for checking permutation-completeness and extending the observation table

\begin{lemma}
	\label{lemma:runtimeClosedConsistent}
	\cite{angluin_learning_regular_sets} Checking the observation table for consistency and closedness can be done in time polynomial in the size of the observation table.
\end{lemma}

The complexity of Algorithm~\ref{algo:extend} is given in the following lemma.

\begin{lemma}
	\label{lemma:runtimeExtend}
	Extending an observation table over arity $k$ to an equivalent observation table of arity $k'>k$ can be done in time polynomial in the size of the observation table and exponential in $k'$. 
\end{lemma}
\begin{proof}
This follows from the runtime of the procedure {\sc Extend}, where we have to perform $|S' \cup S'\cdot \Sigma^{k'}|\cdot E'$ membership queries. As $|S'|=|S|$ and $|E'|=|E|$, the runtime of {\sc Extend} is polynomial in $S$ and $E$ and exponential in $k'$. 
\end{proof}

A $k$-bad-prefix automaton, that is not permutation-complete, accepts one representation of every bad prefix but not every representation of it. Algorithm~\ref{algo:complete} provides a procedure for deciding permutation-completeness.  For $\A = (Q, \Sigma^{k'}, q_0, F, \delta)$ a $k$-bad-prefix automaton and $\varsigma: \{1, \dots, k'\} \rightarrow \{1,\dots, k'\}$, we define $A^\varsigma = (Q, \Sigma^{k'}, q_0, F, \delta^\varsigma)$ where $\delta^\varsigma(s, (a_1, \dots, a_{k'}) = \delta(s, (a_{\varsigma(1)}, \dots, a_{\varsigma(k')}))$. If $\A$ is permutation-complete, then $L(\A^\varsigma) \subseteq L(\A)$ for every $\varsigma$.

\begin{algorithm}
	\caption{\textsc{isComplete}}
	\hspace*{\algorithmicindent} \textbf{Input} $k$-bad-prefix automaton $\A$ over $\Sigma^{k}$ \\
	\hspace*{\algorithmicindent} \textbf{Output} $\A$ permutation-complete, or  \\
	\hspace*{\algorithmicindent} \quad  $w \in \Sigma^{k}$ s.t.~ $w \not \in L(\A)$ \vfill
	\begin{algorithmic}[1]
		\label{algo:complete}
		\FOR {$\varsigma \in \{1,\dots,k\}^{\{1, \dots , k\}}$ }
		\IF {$L(\A^{\varsigma}) \not\subseteq L(\A)$}
			\RETURN $w \in ({L(\A^{\varsigma} - \A)})$
		\ENDIF
		\ENDFOR
		\RETURN Closed
	\end{algorithmic}
\end{algorithm}

\begin{lemma}
	\label{lemma:runtimeCompleteness}
	Checking whether a $k$-bad-prefix automaton $\A$ is permutation-complete for a $k$-safety hyperproperty $\safety$ can be done in time exponential in $k$ and polynomial in $\A$.
\end{lemma}

Building on Lemmas~\ref{lemma:runtimeClosedConsistent}, \ref{lemma:runtimeExtend}  and \ref{lemma:runtimeCompleteness}  the overall complexity of the algorithm L$^*_{Hyper}$ is given by the following theorem.

\begin{theorem}[Learning $k$-Bad-Prefix Automata]
	\label{theo:learning-k-safety_queries}
	Provided a minimal-adequate teacher, L$^*_{Hyper}$ learns a minimal, deterministic, tight and permutation-complete bad-prefix automaton $\A$ for a $k$-safety property $\safety$ over $\Sigma$ in time polynomial in $n$ and $m$, where $n$ is the size $\A$, $m$ is the length of the longest counterexample provided, and in time exponential in $k$.
\end{theorem}

\section{Learning Automata for \hyperltl}
\label{sec:LearningHyperLTL}
The next natural step is to instantiate the $L^*_\textit{Hyper}$ framework to learn automata for \hyperltl.  We dedicate this section to instantiating the \lstarhyper framework for learning minimal permutation-complete automata for universally-safe \hyperltl formulas.

\begin{definition}[Universally-safe \hyperltl formulas]
	A universally-safe \hyperltl formula $\varphi$ is of the form $\forall \pi_1\dots \forall \pi_k.\psi$, where $\psi$ is a safety LTL formula. 
\end{definition}

In the following we show the complexity of deciding membership and equivalence queries for universally-safe \hyperltl formulas.  

\subsection{Deciding membership queries}

\begin{theorem}
	\label{theo:membership}
	Let $T$ be a set of traces, with each trace being of length $n$, and let a universally-safe \hyperltl formula $\varphi = \forall\pi_1\dots \forall \pi_{k_T}.~ \psi$. 
	The problem of deciding whether $T$ is a bad prefix for $\varphi$ can be solved in time polynomial in $n$ and space polynomial in $\abs{\psi}$ and $k_T\cdot\log(\abs{T})$.
\end{theorem}%
\begin{proof}
Deciding whether $t\in \Sigma$ is a bad prefix for a safety \ltl formula $\varphi$ can be done in space polynomial in $\abs{\varphi}$ and time polynomial in $\abs{t}$ by guessing whether $t$ allows an accepting run in the B\"uchi automaton for $\psi$. If no such run is found, then $t$ is a bad prefix for $\psi$. 

For a set $T$  of traces of length $n$ we need to check whether one of the representations is a bad prefix for $\varphi$. 
There are at most $\abs{T}^{k_T}$ many different  representations for $T$ that can be encoded by $\log(\abs{T}^{k_T}) = O(k_T\cdot\log(\abs{T}))$ bits. 
Checking whether $T$ is a bad prefix for $\psi$ can be done in space polynomial in $k_T \cdot \log(\abs{T})$ and  in $\psi$, by guessing the representation and applying the bad prefix check for $\psi$. 
%
\end{proof}

For the purpose of completeness, we present a second algorithm solving membership queries symbolically. To this end, we employ the decidability results of \hyperltl~\cite{finkbeiner_et_al:deciding_hyper}.
In practice, the second algorithm is expected to outperform the first one due to its dependence on SAT solving and the efficiency of state-of-the-art SAT solvers.

\begin{theorem}
	\label{theo:membership_sat}
	Let $T=\{t_1,\dots,t_n\}$ be a set of finite traces and let a universally-safe \hyperltl formula $\varphi = \forall\pi_1\dots \forall \pi_{k_T}.~ \psi$.  
	$T$ is a bad prefix of $\varphi$ if and only if the following \hyperltl formula is unsatisfiable: 
	\vskip 2mm
	\noindent$ \varphi' := \exists \pi'_1 \dots \exists \pi'_n\\ \forall\pi_1\dots \forall \pi_{k_T}.~ \pi'_1 \ge t_1 \wedge \dots \wedge\pi'_n \ge t_n \wedge \psi$
\end{theorem}%
\begin{proof}
	\begin{itemize}
	\item[$(\Rightarrow)$] Let $T = \{t_1, \dots t_n\} \subseteq\Sigma^*$ be a bad prefix of $\varphi$. 
	Thus, for all $T' \subseteq \Sigma^{\omega}$ with $T \le T'$ it holds: $T' \not \models \varphi$.
	Therefore, no traces $\pi'_1, \dots, \pi'_n$ exist that satisfy $\psi$ and  $\varphi'$ is unsatisfiable. \\
	
	\item[$(\Leftarrow)$] Let $T = \{t_1, \dots t_n\} \subseteq\Sigma^*$ be a set of traces and let $\varphi'$ be unsatisfiable.
	We distinguish the following two cases:
	\begin{enumerate}
		
		\item $\varphi$ is unsatisfiable: \\Thus, no set of infinite traces can satisfy $\varphi$ and $T$ is, like every other non-empty set of traces, a bad prefix of $\varphi$.
		\item $\varphi$ is satisfiable: \\Then, the conjunction of $\pi'_1 \ge t_1 \wedge \dots \wedge \pi'_n \ge t_n$ and $\psi$ is unsatisfiable in the context of the given quantifiers. Thus, there does not exists a set of traces $T' \subseteq \Sigma^\omega$ having $n$ traces that extend $t_1, \dots , t_n$ and $T'\models \varphi$. Therefore, $T$ satisfies the definition of a bad prefix of $\varphi$.
	\end{enumerate} 
	\end{itemize}
	
\end{proof}%
\noindent%
According to the results in~\cite{finkbeiner_et_al:deciding_hyper} and since $\varphi'$ is in the bounded $\exists^*\forall^*$ fragment of \hyperltl, \Cref{theo:membership_sat} grants us an algorithm deciding membership queries in space exponential in $\abs{\varphi'} = n \cdot m \cdot \abs{\Sigma} + k_t + \abs{\varphi} $ where $m$ is the length of the traces in $T$.

\subsection{Deciding equivalence queries}
In this section, we focus on the resolution of equivalence queries.
Given an automaton $\A$ and a universally-safe \hyperltl formula $\varphi$, to check whether $\A$ is a bad-prefix automaton for $\varphi$ we need to check that: 
\begin{enumerate}
	
	\item  Every word accepted by $\A$ is a representation of a bad prefix of $\varphi$.
	\item The automaton $\A$ accepts a representation of a bad prefix, for every  set of traces violating $\varphi$.
\end{enumerate}

The next theorem give  an algorithms for deciding problem (1).  Problem  (2) is solve in the  theorem that follows. 

\begin{theorem}
	\label{lem:hyperltl_accept_only_bad_prefix}
	Given a	 universally-safe \hyperltl formula $\varphi =\forall\pi_1\dots \forall \pi_{k_T}.~ \psi$ and a deterministic automaton $\A = (Q, \Sigma^{k_L}, q_0, F, \delta)$. Checking whether every word $w$ accepted by $\A$ is a $k_L$-representation of a bad prefix of $\varphi$ can be done in time polynomial in $\abs{\A}$, exponential in $\abs{\psi}$ and doubly exponential in $k_T$.
\end{theorem}%
\begin{proof}
	We transform $\varphi$ into a \hyperltl formula $\varphi' = \forall\pi_1\dots\forall\pi_{k_T}.~ \psi'(\pi_1,\dots, \pi_{k_T})$ where
	\begin{align*}
	\psi'(\pi_1,\dots, \pi_{k_T}) = \bigwedge_{\varsigma: \{1,\dots,k_T\}\rightarrow {\{1,\dots,k_T\}}} \psi(\pi_{\varsigma(1)}, \dots, \pi_{\varsigma(k_T)})
	\end{align*}
	Note that the trace property described by $\psi'$ is permutation-complete with respect to $\Sigma^{k_T}$, i.e., for all $t \in (\Sigma^{k_T})^\omega$, it holds $t \models_\text{LTL} \psi' \Leftrightarrow \emptyset \models_{\unzip(t)} \varphi$ where the \ltl semantic is adjusted such that $a_{\pi_i}$ holds if $a$ holds in the $i$-th component. 
	The size of $\varphi'$ is exponential in $k_T$ and we can construct a nondeterministic safety automaton $\mc{N}_{\varphi'}$  accepting all infinite sequences that represent a set $T$ of at most $k_T$ traces such that $T\models\varphi$.
	The size of $\mc{N}_{\varphi'}$ is exponential in the size of $\psi'$ \cite{finkbeiner_et_al:algorithms_for_mc_hyper}.

	$\A$ and $\mc{N}_{\varphi'}$ are of different arities, i.e., we need to extend $\A$ to $k_T$
	Let $\A' = (Q', q'_0, \Sigma^k, F', \delta')$ be an automaton, we define the extension $\extend(\A,k')$ of $\A'$ onto $\Sigma^{k'}$ for $k' > k$ as follows: $\extend(\A,k') = (Q', q'_0, \Sigma^{k'}, F',\delta'')$ with  $\delta'(q,\extend(s,k')) = q'$ iff $\delta(q,s) = q'$ for $q, q' \in Q'$ and $s\in \Sigma^{k}$.
	In order to check whether $\A$ accepts any sequence that does not represent a bad prefix, we then construct the nondeterministic product automaton of $\extend(\A,k_T)$ and $\mc{N}_{\varphi'}$, and check the emptiness of the product automaton.  
	%
	%
	%
\end{proof}%

In general, for a \hyperltl formula $\varphi =\forall\pi_1\dots \forall \pi_{k_T}.~ \psi$ we can make the assumption that \lstarhyper never constructs an automaton for an arity larger than $k_T$. 
With this assumption, problem (2) can be solved by checking whether for every set of traces $T$ violating $\varphi$, if $\A$ accepts a representation of a $k_L$-bad prefix of $T$. 

\begin{theorem}
\label{theo:hyperltl_accept_all_bad_prefix}
Let $\varphi =\forall\pi_1\dots \forall \pi_{k_T}.~ \psi$ be a  universally-safe \hyperltl formula and let $\A_L$ be a deterministic bad-prefix automaton over $\Sigma^{k_L}$ for some alphabet $\Sigma$. The problem of deciding whether $\A_L$ recognizes a $k_L$-representation of a bad prefix for every $T \not\in \L(\varphi)$ can be solved in time polynomial in $\abs{\A_L}$,  exponential in $\abs{\psi}$, and space exponential  in $k_T$.
\end{theorem}%
\begin{proof}
For $\varphi$, we can construct a fine, nondeterministic $k$-bad-prefix automaton $\mc{N}_{\varphi} = (Q_{\varphi}, \Sigma^{k_T}, q_{0,{\varphi}}, F_{\varphi}, \Delta_{\varphi})$ for $\psi$ on the adapted alphabet over $\Sigma^{k_L}$ in time exponential in $\abs{\psi}$~\cite{fineautomata}.
In \Cref{theo:equiv}, we provided an equivalence check which solves the above problem with respect to two deterministic $k$-bad-prefix automata. 
However, since only one of the two automata is complemented only $\A_L$ has to be deterministic whereas $\mc{N}_{\varphi}$ can be nondeterministic as well.
Then the complexity bound follows from \Cref{theo:equiv}.

In case the equivalence does not hold, a counterexample of minimal length can easily be produced in the size of the cross-product automaton, similar to \Cref{lem:hyperltl_accept_only_bad_prefix}. 
This counterexample can be reduced to minimal size by repeated deleting of traces as long as the answer to membership with it do not change. This takes linear time in $k_T$ in addition to the complexity of membership queries.
\end{proof}%

\section{Learnability of Hyperproperties}
\label{sec:learnability}
\newtoggle{short}
\toggletrue{short}

We conclude our study with an investigation of the landscape of learnable hyperproperties.
We point out results regarding the theoretical boundaries of learning trace properties, as well as, hyperproperties. 
These findings do not directly affect {\lstarhyper}{,} proposed in the last section, but rather affect possible future extensions.

We start with an elementary result regarding the class of $k$-safety hyperproperties. 
%
This result even holds for $1$-safety hyperproperties, i.e., safety trace properties.

\begin{theorem}
	\label{lem:undecidable_membership}
	Answering membership queries for safety languages is undecidable in general.
\end{theorem}%
\nottoggle{short} {
\begin{proof}	
	Let $\Sigma = \{0,1\}$ and let $\L\subseteq \Sigma^*$ be the binary encoding of the halting problem and let there be an enumeration $M_1,M_2,\dots$ of the language $\L$. 
	Based on $\L$ we  define the following undecidable safety language $\L'$:
	$$ 
	\L' = \filter{0^{\abs{M_i}} \$ M_i \$ \sigma }{ M_i \in \L \wedge \sigma  \in \{0,1\}^\omega} \cup \{0,1\}^\omega
	$$ 
	We show that $\L'$ is a safety language, i.e., every trace not in $\L'$ has a finite bad-prefix. Therefore, let $\tau = \tau_1\tau_2 \dots \in \{0,1,\$\}^\omega$. We consider the following cases: 
	\begin{itemize}
		\item $\tau \in \{0,1\}^\omega:$ $\tau \in \L'$
		\item $\tau \in \{0,1,\$\}^\omega$ with $\#_\$(\tau) = 1:$ Let $\tau = w \$ \sigma \in \Sigma^*\cdot \{\$\}\cdot \Sigma^\omega$. The prefix of $\tau$ of length $2\abs{w} + 2$ is a bad-prefix since its last label was unequal to $\$$.
		\item $\tau \in \{0,1,\$\}^\omega$ with $\#_\$(\tau) > 2:$ The minimal prefix of $\tau$ with $\#_\$ = 3$ is a bad-prefix.
		\item $\tau \in \{0,1,\$\}^\omega$ with $\#_\$(\tau) = 2:$
		\begin{itemize}
			\item $\tau \in \L': \checkmark$
			\item $\tau \not \in \L':$ Choose some $m$ such that $\#_\$(\tau_1 \dots \tau_m) = 1$, $\tau_m =\$$ and $n$ such that  $\#_\$(\tau_{m+1} \dots \tau_{(m+1) + n}) = 1$, $\tau_{(m+1) + n} =\$$. 
			If $m \neq n$, then $\tau_1 \dots \tau_{(m+1) + n}$ is a bad-prefix.
			If $m = n$, then $\tau_{m+1} \dots \tau_{(m+1) + n - 1} \not \in \L$ and thus $\tau_1 \dots \tau_{(m+1) + n}$ is a bad-prefix.
		\end{itemize}
	\end{itemize}
	Thus, $\L'$ is a safety-language. $\L'$ is undecidable since $\L$ was chosen to be the binary encoding of the halting problem.

		Consider the language $\Sigma = \{0,1\}$. 
	Let $\L$ be the binary encoding of the halting problem and let there be an enumeration $M_1,M_2,\dots$ of the language $\L$. 
	
	We will define a suffix-closed undecidable language based on $\L$, since the \TODO{the complement is not really a safety language, but the infinte complement, think of a better description.}
	Define the following language $\L'$ as
	$$ \L' = \filter{0^i\#M_i\# \{0,1\}^*}{ M_i \in \L} ,$$

	where $\$ \not \in \Sigma$. The constructed language $\L'$ is suffix-closed, i.e., $\forall y\in \{0,1,\#\}^*.~(\exists x\in \L. ~ x \le y) \Rightarrow y \in \L'$. 
	Thus, $\L'$ is a bad-prefix language and $\overline{\L'}$ is a safety-language. 
	
	The language $\L'$ is undecidable because $\L$ was chosen to be the binary encoding of the halting problem. 
	Thus, $\overline{\L'}$ is undecidable as well, and deciding membership queries for $\overline{\L'}$ is undecidable.
	
\end{proof}%
\noindent
}

	The previous theorem can be proved by a reduction to the halting problem. Therefore, one constructs an undecidable safety language. Membership for this language is undecidable. Thus, the theorem follows.
Since the class of $1$-safety hyperproperties is a equal to the set of  safety languages ~\cite{clarkson_et_al:hyperproperties}, the subsequent corollary immediately follows from \Cref{lem:undecidable_membership}.

\begin{corollary}
	\label{cor:undecidable_membership_k_safety}
	Membership queries for $k$-safety hyperproperties are, in general, undecidable.
\end{corollary}%

Despite the undecidability result, there are many important classes of safety languages for which membership queries are decidable. One important example is the class of safety languages defined in \ltl.
Here the problem can be decided by constructing the conjugation between the \ltl formula at hand and the given prefix, expressed in \ltl. Afterwards checking the satisfiability of the obtained formula decides the membership~\cite{DBLP:conf/ecai/Li0PVH14}.
In addition, equivalence queries for the safety fragment of \ltl can be decided: 
Let $\A$ be a deterministic bad-prefix automaton and $\psi$ an \ltl formula. 
We can construct a deterministic bad-prefix automaton $\A_\psi$ for $\psi$ accepting all bad-prefixes of $\psi$~\cite{kupferman_et_al:model_checking_safety_properties}. $\A$ is a bad-prefix automaton for $\psi$ if $\A$ and $\A_\psi$ accept the same language, which can be easily checked.
Hence, the decidability of queries of safety languages defined in \ltl is completely solved. Thus, automated approaches using the learning of safety properties in \ltl are applicable, in general. 

Looking at our learning framework for hypersafety properties expressed in \hyperltl, it becomes of interest to ask, to what extend queries can be decided? 
Regarding equivalence queries for \hyperltl, we obtain the following negative result. The proof is independent of the representation model for \hyperltl.

\begin{proposition}
	Equivalence queries are undecidable for the full class of hyperproperties expressed in \hyperltl.
\end{proposition}%
\begin{proof}
	Assume that we can decide membership and equivalence queries for $\varphi$. Then, we can answer the satisfiability of $\varphi$ by asking whether $\varphi \equiv \textsf{False}$. This contradicts the undecidability of HyperLTL-SAT~\cite{finkbeiner_et_al:deciding_hyper}. 
\end{proof}%

Thus, it becomes important to consider subclasses of \hyperltl, like the universal-safe \hyperltl fragment. 
In order to use such subclasses in automated learning environments, it is necessary to decide whether a \hyperltl formulas is an element of this subclass.
For \ltl, the corresponding question---whether a formula expresses a safety property---is decidable ~\cite{DBLP:journals/ipl/MareticDB14}. It is even possible to decompose the formula into its pure safety and liveness parts.
Clarkson and Schneider showed that such a separation into a pure hypersafety and a pure hyperliveness part always exists for hyperproperties~\cite{clarkson_et_al:hyperproperties}. Thus, it looks promising that such a separation can be computed.
In the following theorem we reduce the satisfiability of \hyperltl to deciding whether a given \hyperltl formula describes a $k$-hypersafety property.
Thus, according to the undecidability result for the satisfiability of \hyperltl, our problem is undecidable \cite{finkbeiner_et_al:deciding_hyper}.


%
%

\begin{theorem}
\label{theo:hyperltlissafety}
	Whether a \hyperltl formula $\varphi$ expresses a $k$-safety hyperproperty is undecidable in general. 
\end{theorem}%
\begin{proof}
	We start by observing that every unsatisfiable \hyperltl formula describes the hyperproperty $\L(\varphi)=\emptyset$ and $\L(\varphi)$ is $k$-safety for all $k \in \NN$. 
	
	Let $\varphi = Q_1 \pi_1. \dots Q_n \pi_n.~ \psi$ be \hyperltl formula over the alphabet $\Sigma$. We construct the formula $\varphi' = Q_1 \pi_1. \dots Q_n \pi_n.~ \psi \wedge \finally a_{\pi_1}$  over the alphabet $\Sigma' = \Sigma \mathbin{\dot{\cup}} \{a\}$, i.e.,  $a \not \in \Sigma$	. 
	We will prove that $\varphi$ is unsatisfiable if and only if  $\L(\varphi')$ expresses a $k$-safety hyperproperty. Therefore we distinguish the following two cases:
	\begin{itemize}
		\item $\varphi$ is unsatisfiable:  $\varphi'$ is unsatisfiable as well. Thus, $\varphi'$ is $k$-safety for any $k\in \NN$.
		\item $\varphi$ is satisfiable: Let $T \subseteq \Sigma^\omega$ be a set of traces such that $\emptyset \models_T \varphi$. 
		Construct $T' \subseteq {\Sigma'}^\omega$ $ T' = \{ t' \mid \exists t \in T. ~ t'\equiv_{\Sigma} t   \wedge \forall i \in \NN.~ t'[i] \models a\}$. 
		Since $T\models \varphi$ and $T'\models \forall \pi. a_\pi$ it follows: $T' \models \varphi'$. 
		Thus, $\varphi'$ is satisfiable and $\L(\varphi')$ is not a $k$-safety hyperproperty since the extra conjunct enforces every or some trace, depending on $Q_1$, to satisfy eventually $a$, a liveness requirement.
	\end{itemize}
\end{proof}%
\noindent
The same proof provides that checking if a \hyperltl formula describes a safety hyperproperty is undecidable.

In contrast to the negative decidability results, we give a decision procedure for deciding whether a formula is $k$-safe for the important $\forall$-fragment of \hyperltl. This, is a superset of the universal-safe formulas.

\begin{theorem}
	Let $\varphi = \forall \pi_1 \dots \forall \pi_k.~ \psi$ be a \hyperltl formula. 
	It is decidable if $L(\varphi)$ is a $k$-safety hyperproperty, in space polynomial in $\abs{\psi}$ and space exponential in $k$. 
\end{theorem}%
\begin{proof}
	First, note that for \hyperltl formulas in the $\forall$-fragment, the described hyperproperty is safe if and only if it is $k$-safe.
	By the definition of $k$-safety hyperproperties, $L(\varphi)$ is $k$-safe if and only if 
	$$
		\forall T.~ \left[T \not\models \varphi \Rightarrow  \left[\exists T' \le T.~ \abs{T'} \le k \wedge \forall \tilde{T'} \ge T'.~ \tilde{T'} \not \models \varphi\right]\right]
	$$
	Our first claim is: $\varphi$ is safe for all sets of size at most $k$ if and only if $\varphi$ is safe. 
	\begin{itemize}
	\item[$(\Leftarrow)$] Let $T$ be some set of traces that violates $\varphi$. 
	There exists a set $T' \le T$ with $\abs{T'}\le k$ and thus $\varphi$ must reject every set of traces extending $T'$ including those of size $k$.
	Hence, $\varphi$ is safe for traces of size at most $k$.
	
	\item[$(\Rightarrow)$] Let $\varphi$ be $k$-safe for all sets of traces of size at most $k$. Then, a set of traces violating $\varphi$ must have a bad-prefix $T'$ of size at most $k$.
	Further, every set of traces extending $T$ has this prefix $T'$ and thus $\varphi$ is safety.
	\end{itemize}	
	Further, we claim that $\varphi$ is safe for a set of traces $T$ of size $i$ if and only if the following formula is safe with respect to LTL semantics:
	\begin{align*}	
		\psi_i(\pi_1, \dots , \pi_i) := \bigwedge_{\varsigma: \{1,\dots, k\} \rightarrow \{1,\dots, i\}} \psi(\pi_{\varsigma(1)}, \dots , \pi_{\varsigma(k)})
	\end{align*}
	We rephrase the foregoing claim as follows: for all $m \le k$ and $T = \{t_1,\dots, t_m\}$: 
	$$
		 T \models \varphi\text{ if and only if }(t_1, \dots, t_m) \models \psi_m
	$$
	\begin{itemize}	
		\item[$(\Rightarrow)$]	
			Following the semantic of \hyperltl $T \models \varphi$ implies:
			For all $\Pi$ with $\textsf{Traces}(\Pi) \subseteq T.~ \Pi \models_\emptyset \psi$. 
			Thus, $(t_1, \dots, t_m) \models \psi_m(\pi_1, \dots, \pi_m)$. \\
	
		\item[$(\Leftarrow)$]	
			$(t_1, \dots, t_m) \models \psi_m$ implies that for all $\varsigma:\{1, \dots, k\}\rightarrow \{1, \dots, m\}$: $(t_{\varsigma(1)}, \dots, t_{\varsigma(k)})\models\psi(\pi_1, \dots , \pi_k)$.
			Hence, for all trace assignments $\Pi$ with $\textsf{Traces}(\Pi) \subseteq T.~ \Pi \models_\emptyset \psi$.
			And thus by definition $T \models \varphi$. 
	\end{itemize}
	Hence, it follows $\varphi$ is safe for trace sets of size up to $m$ if and only if $\psi_{m}$ is safe under the semantics of \ltl.
	Therefore, deciding whether $\psi_{k}$ is reduced to checking \ltl safety, which can be computed in space polynomial in $\abs{\psi_k}$~\cite{DBLP:journals/ipl/MareticDB14}.
	The result coincides with whether $\varphi$ is a $k$-safety hyperproperty.
\end{proof}

In the last two sections, we have shown that for \hyperltl formulas in the universal-safe fragment, we can both decide if a given formula belongs to the fragment and decide membership queries. For the safe hyperproperties in the $\forall^*$-fragment of \hyperltl, we can decide whether a formulas belongs to the fragment, but the decidability of the queries is open. In general, equivalence queries are undecidable. 


\IGNORE{

\section{Safety Hyperproperties in \hyperltl}
During the application of the learning framework to $k$-safety hyperproperties we restricted ourselves to hyperproperties that are expressed in the $\forall$-fragment of \hyperltl. 
It is therefore of interest to investigate the set of hypersafety properties expressible in the $\forall$-fragment compared to those expressible in the full logic \hyperltl. 
At first, we show that there are safety hyperproperties not expressible in the $\forall$-fragment, this is mainly based on the restriction of formulas in the $\forall$-fragment to only being able to respect $k$-traces at a time, where $k$ is the number of $\forall$ quantifiers used. 

\begin{lemma}
	There exist safety hyperproperties that are expressible in \hyperltl but not in the $\forall$-fragment of hyperproperties. 
\end{lemma}%
\begin{proof}
	\TODO{find a safety hyperproperty in \hyperltl that is not $k$-safe}

	\QED
\end{proof}%
\noindent
It remains an open question whether the set of $k$-safety hyperproperties expressible in \hyperltl coincides with the set of $k$-safety hyperproperties expressible in the $\forall$-fragment.

\IGNORE{
\section{Safety in \hyperltl is not Expressible in the Forall Fragment}

$$\forall \pi. \exists \pi'. a_{\pi} \until (\neg a_{\pi} \wedge \globally a_{\pi})$$

$$\forall \pi.  \forall \pi'' \forall \pi' \exists\pi''' \forall \pi''''. \ \globally\neg a_{\pi} \lor \bigg[ (a_{\pi''} \land \lnext \globally \neg a_{\pi''}) \release (\neg a_{\pi''}) \land \Big[\neg a_{\pi'} \lor \big( \finally a_{\pi'''} \land \globally (a_{\pi'''} \Rightarrow \neg \lnext a_{\pi''''} ) \big)\Big]\bigg]$$
idea:
\begin{itemize}
	\item on no trace a holds, or
	\item on every trace a holds at most one time and not the infinite structure of traces that `push' each other holds.
\end{itemize}

This property is not expressible in the pure forall fragment of \hyperltl, therefore, there exists hypersafety properties in \hyperltl outside of the forall fragment. 

Remark: This is not as hypersafety property in the classic notion, as the bad prefixes are of infinite size. 
}
}

\section{Related Work}

Most verification techniques for $k$-safety hyperproperties  are based on self-composition \cite{Barthe:2011:SIF:2139690.2139694,10.1007/978-3-642-35722-0_3}. Self-composition enables the use of standard techniques for information flow policy verification, such as program logics and model checking. 
Automata-theoretic approaches for the verification of information-flow policies include model checking algorithms, such as for \hyperltl \cite{finkbeiner_et_al:algorithms_for_mc_hyper}, for Mantel's Basic Security Predicates BSPs \cite{DSouza:2011:MTI:2590694.2590698}, and for epistemic logics \cite{DBLP:phd/basesearch/Balliu14}. There are also related algorithms for synthesis~\cite{hypersynthesis}, satisfiability~\cite{finkbeiner_et_al:deciding_hyper}, and monitoring~\cite{DBLP:conf/csfw/AgrawalB16,finkbeiner_et_al:monitoring_hyperproperties,10.1007/978-3-030-17465-1_7}. Typically, these approaches rely on automata constructions for trace properties. For example, in automata-based HyperLTL monitoring~\cite{finkbeiner_et_al:monitoring_hyperproperties}, an automaton is constructed for the
underlying LTL formula over an indexed set of atomic propositions. During monitoring, this automaton is then applied to multiple combinations of the observed traces by instantiating the indices in all necessary permutations.
None of these approaches define a canonical representation for hyperproperties. 

There is a rich body of work on learning from examples, ranging from learning automata \cite{angluin_learning_regular_sets,Angluin:1988:QCL:639961.639995,Farzan:2008:EAC:1792734.1792738} to approaches for specification mining on systems \cite{vanLamsweerde:1998:IDR:297569.297578,Fern:2004:LDC:3037008.3037033,Cobleigh:2006:UGC:1181775.1181801}. Further algorithms have been presented for specification mining of information-flow polices~\cite{Clapp:2015:MME:2771783.2771810,Livshits:2009:MSI:1543135.1542485}. Approaches for learning specifications for monitoring malicious behavior were presented in \cite{Christodorescu:2007:MSM:1287624.1287628}. Our learning approach provides, to the best of our knowledge, the first general framework for learning information-flow policies.

\section{Conclusion}

We have presented the first canonical representation  for $k$-safety-hyperproperties. We introduced automata for representing $k$-safety hyperproperties and gave algorithms for constructing permutation-complete automata for such hyperproperties.  We also presented the learning framework \lstarhyper that can be used to learn minimal permutation-complete automata for $k$-safety hyperproperties and gave an instantiation for  \hyperltl. 
The advantage of the algorithm is that it allows us to interactively learn monitors for information-flow polices and automatically construct efficient monitors from \hyperltl specification. It further allows for the simplification of manually specified $k$-safety-properties and for automatic equivalence checks between hyperproperties. 

As a natural next step, we plan to implement the learning algorithm for \hyperltl and investigate further classes of \hyperltl formulas beyond the universally-safe fragment. Moreover, we plan on investigating further possible canonical representations for $k$-safety hyperproperties. For example, instead of looking for automata that accept the representations of all bad prefixes up to the arity $k$, we can change the definition of tightness to mean representations of only minimal bad prefixes. One advantage of this definition is,  that in the case of \hyperltl, we can skip the  additional expensive check that the conjecture automaton is vertically tight. This however, comes with the trade off, that 
 the learner now  has to extract minimal bad prefixes out of the counterexamples,  which adds exponential costs. 


\balance
\bibliography{bibliography}
\bibliographystyle{IEEEtran}

%

\end{document}